\documentclass{article}
\usepackage{amsmath}
\usepackage{amsfonts}
\usepackage{amssymb}
\usepackage[utf8]{inputenc}
\usepackage[T1]{fontenc}
\usepackage[colorlinks=true]{hyperref}
\hypersetup{%
  colorlinks = true,
  linkcolor  = black
}
\usepackage{graphicx}
\usepackage[dvipsnames]{xcolor}
\usepackage{authblk}
\usepackage{soul}
\usepackage{cite}

\usepackage{amsthm}

\newtheorem*{LMP}{Local Mach Principle (colloquial form)}
\newtheorem*{LMPt}{Local Mach Principle (technical form)}
\newtheorem*{LUQT}{Local Universality of Quantum Theory}
\newtheorem{example}{Example}

\newtheorem{observation}{Observation}

\newtheorem{thm}{Theorem}

\newtheorem{conjecture}{Conjecture}
\newtheorem{prop}{Proposition}
\newtheorem*{ass}{Assumption}

\def\be{\begin{equation}}
\def\ee{\end{equation}}
\def\ba{\begin{eqnarray}}
\def\ea{\end{eqnarray}}

\newcommand\q{\quad}

\definecolor{mm}{rgb}{0.57, 0.36, 0.51}

\newcommand{\ca}{\mathcal A}

\newcommand{\cb}{\mathcal B}

\newcommand{\cg}{\mathcal G}

\newcommand{\cm}{\mathcal M}

\newcommand{\dd}{\mathrm{d}}

\usepackage{geometry}
\title{A local quantum Mach principle and the metricity of spacetime}

\author[1,2]{Philipp A.\ H\"ohn\thanks{\texttt{p.hoehn@univie.ac.at}}}
\author[1,5]{Markus P.\ M\"uller\thanks{\texttt{markusm23@univie.ac.at}}}
\author[3]{Christian Pfeifer\thanks{\texttt{christian.pfeifer@ut.ee}}}
\author[4,2]{Dennis R\"atzel\thanks{\texttt{raetzeld@physik.hu-berlin.de}}}
\affil[1]{\small Institute for Quantum Optics and Quantum Information, Austrian Academy of Sciences,\newline Boltzmanngasse 3, 1090 Vienna, Austria}
\affil[2]{\small Vienna Center for Quantum Science and Technology (VCQ), Faculty of Physics, University of Vienna, Boltzmanngasse 5, 1090 Vienna, Austria}
\affil[3]{\small Laboratory of Theoretical Physics, Institute of Physics, University of Tartu, W. Ostwaldi 1, 50411 Tartu, Estonia}
\affil[4]{\small Institut f\"ur Physik, Humboldt-Universit\"at zu Berlin,  12489 Berlin, Germany}
\affil[5]{\small Perimeter Institute for Theoretical Physics, Waterloo, ON N2L 2Y5, Canada}

\date{}

\begin{document}

\maketitle

\begin{abstract}

We revisit the old question of what distinguishes the formulation of spacetime geometry
in terms of a Lorentzian metric physically from more general geometric
structures, such as, e.g., general dispersion relations.
Our approach to this question is operational and leads us to also revisit the notion of local inertial frames, arising in operational formulations of the equivalence and Mach's principle, both of which can be interpreted in generalized geometries. We extend the notion of inertial laboratory frames by taking serious that all matter inside the lab is fundamentally quantum and considering how it may or may not couple to the quantum gravitational degrees of freedom generating the ambient effective spacetime structure. This revolves around the more specific question of which structures an agent inside the inertial laboratory has available to operationally define the orientation of their reference frame, an aspect on which both the equivalence and Mach's principle impose no further restrictions. We then contemplate the situation of a completely inertial laboratory, which, in terms of the quantum matter experiments inside it, is not only isolated from any matter outside it, but also from a direct coupling to effective quantum gravitational degrees of freedom. We formulate this in the form of what we shall term a {\it local Mach principle} (LMP): a local inertial laboratory has to be self-sufficient, so that an agent can only resort to relations among the quantum matter systems inside it to self-generate any reference structures relative to which to orient their frame. The transformations between different frame orientations thereby originate in the  local quantum matter structures. Combining this with dispersion relations leads to various non-trivial compatibility conditions on the spacetime structures encoded by them. This permits us to formulate additional operational assumptions under which the LMP singles out Lorentzian metric spacetimes within generalized geometries defined by dispersion relations.

\end{abstract}

\section{Introduction}

The geometry of spacetime plays a fundamental role in our understanding of physics. It defines observer reference frames, encodes the relations between their respective descriptions of the physics, and determines the causal structure, all of which are indispensable ingredients for consistently describing the operational experiences of observers in spacetime. The geometry, furthermore, provides a description of gravity, its dynamics and the coupling of all matter to spacetime. In general relativity, our current best working theory of gravity, this manifold role of  geometry  is conveniently realized in terms of a metric with Lorentzian signature. Observer reference frames, their orientation and mutual relations are realized in terms of frames, which are oriented and normalized according to the metric and related by appropriate transformations of the metric.
The causal structure is encoded in the distinction between causal (timelike and lightlike) and spatial directions, which determine when one event can influence another. The gravitational dynamics, on the other hand, is governed by the Einstein equations, while any additional matter dynamics is subject to the respective field equations, which also couple to the metric.

But why should the geometry of spacetime be described in terms of a Lorentzian metric? There exist more general geometric structures, which too could, in principle, realize the manifold role that the geometry of spacetime assumes in our understanding of physics. The viability of more general spacetime structures is, of course, contingent on whether they are not only mathematically consistent, but also in harmony with observations. This is a long debated issue that relates also with efforts exploring whether more general geometric structures could even be capable of explaining observations, such as apparent dark matter or dark energy phenomena \cite{Riess:1998cb,Peebles:2002gy,Corbelli:1999af,Clowe:2006eq}, on which general relativity seems to fail. Similarly, generalized geometric structures also arise in approaches to quantum gravity phenomenology \cite{AmelinoCamelia:2008qg,Hossenfelder:2012jw,Liberati:2013xla,Mattingly:2005re}. It is thus pertinent to better understand Lorentzian metrics in a wider context of geometric possibilities from a physical point of view. Specifically, answering the above question amounts to investigating the physical consequences of generalized geometries and to ask for {\it physical} properties that single out Lorentzian metric spacetimes within a large class of them.

There is, of course, a vast body of work on generalizations of the spacetime structure of general relativity. To name just a few classical generalizations, there are scalar-tensor-vector theories \cite{Heisenberg:2018mxx,Sotiriou:2008rp, Capozziello:2011et,Lovelock:1971yv, deRham:2014zqa}, teleparallel or Poincar\'e gauge theories of gravity \cite{Clifton:2011jh,Cho:1975dh,Ferraro:2006jd,Hehl:2012pi}, Robertson-Mansouri-Sexl gravity \cite{Robertson1949,Mansouri1977,Mansouri1977a,Mansouri1977b}, 
or Finsler and Cartan geometry \cite{Bogoslovsky:1999pp,Pfeifer:2011xi,Gielen:2012fz,Javaloyes:2018lex,Hohmann:2015pva,Hohmann:2013fca}, etc. A very general approach, which not only encompasses many classical spacetime structures, incl.\ various of the above generalizations, but also naturally connects with effective approaches to quantum gravity phenomenology \cite{AmelinoCamelia:2008qg,Liberati:2013xla,Mattingly:2005re,Gubitosi:2013rna,Barcaroli:2017gvg}, is defining spacetime structure through dispersion relations \cite{Raetzel:2010je,Rubilar:2007qm,Punzi:2007di,Barcaroli:2015xda,Schuller:2016onj}. This is the approach we shall adopt below when attempting to elucidate characterizing physical features of Lorentizan metric spacetimes.

Historically, the emergence of Lorentzian metric spacetime structure can be traced to the realization that observer models must be consistent with Maxwell's electrodynamics and especially with the symmetries of light propagation \cite{ThePrincipalOfRelativity}. In the same vein, but from a more modern point of view, one can study the constraints that physically viable matter field theories and their dynamics impose on generalized geometric structures \cite{Raetzel:2010je,Punzi:2007di,Rubilar:2007qm,Gurlebeck:2018nme,Schuller:2016onj}.

In this work, however, we will pursue an alternative strategy for seeking characterizing physical properties of Lorentzian metrics within spacetime structures defined through dispersion relations. Instead of presupposing a specific matter content and dynamics, we shall follow an operational, principles based approach, imposing conditions  on  what an observer can or cannot do in their local laboratory and studying which consequences this has for spacetime structure \cite{smolin1986nature,Hoehn:2014vua,Hoehn:2017gst,Hardy:2018kbp}. Such an operational approach has also played a pivotal role in the development of general relativity. 
The point we exploit is: two key  principles underlying general relativity do not actually imply or rely on a metric spacetime structure and can also be interpreted more generally. 

For example, the Einstein equivalence principle (EEP) states that the (in both a spatial and temporal sense) local physics in a sufficiently small freely falling frame is indistinguishable from that in an inertial frame in empty space. Usually, this is interpreted as meaning that to an observer in free fall, spacetime and the physics in it will locally look like Minkowski space. But there is no a priori reason for doing so and `empty space' can be incarnated in many different geometric structures. 

Another key impetus for general relativity was Mach's principle, which concerns what actually defines an inertial or free fall frame (as used in the EEP); essentially, it states that an inertial frame is determined with respect to all dynamical degrees of freedom of the universe \cite{Rovelli:2004tv, mercati2018shape,Barbour295}. As such, we shall term it the {\it global Mach principle} (GMP). In particular, it also entails a {\it global relationalism}: physical systems and frames are not localized and oriented relative to some absolute spatiotemporal structure, but relative to other dynamical degrees of freedom in the universe. The ensuing global relationalism has essentially lead to the gauge symmetry of general relativity, namely its diffeomorphism symmetry \cite{Rovelli:2004tv}, which acts globally on spacetime. It is clear that the GMP too, as a physical statement, does not actually refer to or entail metric structures and can be realized in generalized spacetime structures.

Our line of attack for characterizing Lorentzian metric structures among dispersion relations is the following observation. Given a free fall frame, in accordance with the EEP and GMP, these two principles entail that an observer inside it must orient their frame relative to dynamical degrees of freedom, but impose no further operational restrictions on the nature of these degrees of freedom. In particular, the GMP is silent on whether these correspond to distant systems (e.g., the fixed stars) or the physics in the vicinity (e.g., the solar system or even the physics in the local laboratory, based on local fields, like the electromagnetic one). This vast freedom and the absence of operational constraints on frame orientation are related to the fact that the conjunction of the EEP and GMP can be realized in many generalized geometries.

One of our aims is to remedy this situation by imposing further operational conditions in the form of an additional principle. We will introduce below what we call a {\it local Mach principle} (LMP), which will take serious that an observer inside a local inertial laboratory, complying with the EEP and GMP, will experience the physics inside it to be isolated from the rest of the universe, in the sense that no local matter experiments will inform the observer about the outside of the lab. As such, the LMP will require a local inertial frame to be self-sufficient and to self-generate its own reference structure from only the matter physics inside it that the observer can also, in principle, directly control. But all matter is fundamentally described by quantum theory. Hence, the LMP will imply a {\it local quantum relationalism}: physical systems in local inertial laboratories can only be localized and oriented relative to other local quantum matter degrees of freedom. This is an extension of the usual notion of local inertial laboratory, as it posits that the quantum matter physics inside is not only isolated from interactions with the matter outside, but also from net interactions with the quantum gravitational degrees of freedom, which produce the ambient effective spacetime structure.

In analogy to how the global relationalism implied by the GMP led essentially to the {\it globally} acting diffeomorphism symmetry, this local quantum relationalism will operationally constrain the {\it local} symmetry of spacetime. Indeed, if the different local frame orientations are defined relative to the quantum matter in the laboratory, the transformations between them must emerge from quantum structures. 
The crucial point is that these local symmetries emerging from quantum matter structures alone (which experience a vanishing net interaction with effective quantum gravity degrees of freedom) must be compatible with the dispersion relation defining the effective ambient spacetime structure. This non-trivial condition will lead to various operational consequences, which we will detail below, and in one formulation we will find that it singles out Lorentzian metric spacetimes. As an aside, we will also derive several new results about dispersion relations and, in particular, about when they feature local and linear symmetries.

The rest of this article is organized as follows. In sec.\ \ref{sec_oper}, we begin with very general operational considerations on an observer in a local quantum laboratory residing in some  spatiotemporal environment. We start by discussing in detail the motivations for a local quantum Mach principle and in particular discuss explicitly the assumptions which are needed for its formulation in sec.\ \ref{sec_lmp}. We state the LMP in a colloquial form which we then make precise in a technical mathematical formulation from which then deduce successively that the set of transformations among different local frame orientations must contain at least the rotation (Observation 1) or Lorentz (Observation 2) group, respectively. We illustrate the validity of Observation 2 and further background assumptions in Sec.\ \ref{SubsecExMinkowski}. In sec.\ \ref{sec_wald}, we reverse our perspective and ask how a given `local symmetry' group of spacetime acts on the quantum physics in a local laboratory. Subsequently, in sec.\ \ref{sec:relloc}, we argue operationally that the frame transformations from sec.\ \ref{sec_lmp} and the `local isometries' must be isomorphic if the LMP holds and in sec.\ \ref{sec_nolmp} we discuss some possible operational consequences of a violation of the LMP. Having clarified the consequences of a LMP we consider spacetime geometries defined by dispersion relations in sec.\ \ref{sec_finsler} which we want to confront with the LMP. We begin by technically introducing dispersion relations as Hamilton functions on the cotangent bundle of spacetime in sec.\ \ref{sec:DispHam} where we also clarify the notion of spacetime symmetries on the basis of dispersion relations. In sec.\ \ref{sec:ObsDispRleations} we define observers and their mass shell encodings, where the equivalence classes of all observers with identical mass shell encodings links observer transformations and the local symmetries of spacetimes to the operational groups defined in context of the LMP. In sec.\ \ref{sec_3g} we demonstrate how the appearing groups are related to each other before we finally conclude on which kind of dispersion relations, defined in terms of Hamilton functions, yield spacetime structures compatible with the LMP in sec.\ \ref{sec:implLMP}.  Before we conclude, we demonstrate our findings with some illustrative examples in sec.\ \ref{sec:ex}.

For better orientation, we provide a table of contents.
\tableofcontents

\section{A local Mach principle and frame orientations}\label{sec_oper}

\subsection{Motivation for a local quantum Mach principle}\label{sec_0lmp}

To put our work into context, it is worthwhile to revisit a few notions pertaining to local inertial frames and to separate three questions:
\begin{itemize}
\item[(a)] What is a local inertial frame?
\item[(b)] What determines whether a local frame is inertial or not?
\item[(c)] Given a local inertial frame, what structures does an agent inside it have available to specify the orientation of this frame?
\end{itemize}

In the sequel, we will assume a local laboratory frame to reside in an effective spacetime structure that, while emerging as some suitable large-scale (coarse-graining) limit of quantum gravity, can be treated in {\it classical} terms. In this work, we will use the term `frame' in an operational sense, i.e.\ as a genuine {\it physical} local laboratory system, as it also appears in the colloquial formulations of the EEP and GMP.  We shall not yet specify further at this stage what the effective classical spacetime structure is. 
In fact, it will be precisely our aim to impose a few operational conditions on inertial observers and to see which spacetime structures will be compatible with them. In particular, we will {\it not} assume the spacetime structure to be defined by a Lorentzian metric. Instead, we wish to formulate operational statements that characterize Lorentzian metric spacetimes within a large class of effective spatiotemporal structures that need not be defined by metrics at all.

Within this effective spacetime structure, we take the conceptual answers to these questions to be given as follows. These answers also specify some of the basic properties which we assume the effective spacetimes to feature. Subsequently, we will technically clarify these structures.\\~

\noindent\textbf{(a)} A local inertial frame is a frame in which (a version of) {\it Einstein's equivalence principle} holds. That is, colloquially, it is a 
sufficiently small (spatially and temporally local), freely falling laboratory in which the local physics is indistinguishable from that in an inertial frame in empty space. Being a freely falling frame in gravity (or an inertial frame in empty space), the net force exerted by other matter onto this frame is zero. More precisely, the net matter, (i.e.\ non-gravitational) interactions of the frame and anything inside it with the matter outside is zero; in terms of non-gravitational interactions, this frame and the physics inside it are isolated from the rest of the matter in the universe. Furthermore, being (locally) inertial, the agent inside the frame can also not detect any special directions of acceleration due to gravitational interaction with the outside. Accordingly, by local experiments with the matter inside the frame, the agent cannot detect whether there is matter outside of it.

If all this physics were classical and the spatiotemporal structure understood in terms of Lorentzian metrics, this would be just the standard interpretation of the EEP and we would have nothing new. However, here we interpret the previous paragraph much more generally, in line with our resort to effective spacetime structures. Firstly, by `empty space' we do not necessarily refer to a vacuum state in Minkowski space. For example, more generally, it could be a vacuum state in some flat Finsler spacetime \cite{Pfeifer:2014eva,Minguzzi2016}. Similarly, by free fall frame, we do not necessarily refer to a timelike geodesic in a Lorentzian metric spacetime. Instead, it could be a timelike geodesic in a curved Finsler geometry. In fact, we will not even work with Finslerian spacetime structures here, but later use dispersion relations to define spatiotemporal structures and so the latter could be even more general and we will not refer to any geodesic principle in this work.  
 
Secondly, an agent can only resort to {\it matter} experiments inside the lab to test its (locally) inertial nature. But since all matter is fundamentally quantum, we henceforth accept the following:
\begin{LUQT}
	All the matter physics in a local inertial frame is fundamentally described by unitary quantum theory.
\end{LUQT}
``Quantum theory'' here is an extremely general notion: it refers to the general framework of Hilbert spaces and operator algebras without specifying further details. As such, it encompasses (relativistic) quantum mechanics and also quantum field theory. Unitarity implements the observation above that local inertial frames are isolated from the rest of the matter in the universe: the frame and its matter content evolve according to a fixed Hamiltonian which, in particular, is not time-dependent (over the timescale in which the physics inside the lab are indistinguishable from that in flat space). Similarly to the proposal in \cite{smolin1986nature}, there are thus no dissipative quantum field effects that the agent could locally detect. Hence, inside a local inertial laboratory, the appropriate flat space quantum theory holds.

While we assumed the local frame to be sufficiently small for it to satisfy the EEP, we will also assume it to still be large enough so that it is sufficiently classical and we {\it may} (i) treat the observer inside it, for operational purposes, as a classical agent, and (ii) describe its orientation in terms of standard classical frame (tetrad) vectors on a manifold. \\~

\noindent\textbf{(b)} According to the {\it global Mach principle}, the entire dynamical content of the universe determines whether a local frame is inertial. Here, in analogy to general relativity, we assume that `spacetime tells matter how to move; matter tells spacetime how to shape'. That is, in extension of the Einstein field equations, we assume the effective dynamics of spacetime and matter to be intimately linked, so that both gravitational and matter degrees of freedom ultimately determine what free fall in these effective spacetime structures means. In particular, given that also the effective spacetime structure is assumed to emerge from some quantum gravity degrees of freedom, it is purely quantum degrees of freedom that determine whether a frame is inertial or not.\\~

\noindent\textbf{(c)} Addressing the question of what structures an agent in a local inertial frame can exploit to orient their frame is one of the main points of this work. First, we note that allowing the agent to look outside their lab to use galaxies or stars in their vicinity to determine an orientation of their frame would lead, strictly speaking, to an inconsistency. Either the local frame is no longer truly inertial because of the light interaction with the outside, or we would have to consider it in such an idealized fashion 
 that it neither back-reacts on spacetime nor that the light signals an agent receives or sends out affect the inertial nature of the frame. That is, for all practical purposes, the frame would be external to spacetime and just `painted onto' it. 

Such an idealization is incompatible with our operational approach here, which takes serious that a local inertial frame is a physical system {\it in} spacetime (see also \cite{Hoehn:2017gst}). We will henceforth also be strict about the notion of locally inertial (see (a)). What are then the structures that an internal agent can exploit to define an orientation of their frame?

As we have seen, there are no non-gravitational interactions between the matter inside and outside of the frame and the agent inside the local inertial lab is also unable to detect any gravitational effects from the matter outside on the matter inside. In that sense, despite gravitational interactions, the matter inside is isolated from that outside and so the latter cannot provide any help in orienting the frame. The only remaining physics that {\it could} then offer non-trivial structure in the local inertial laboratory, which the agent could exploit to define an orientation of the frame is the local matter quantum physics, or the effective spacetime structure inside it. But the agent cannot control the latter because, thanks to (b), it will also be degrees of freedom, incl.\ quantum gravitational ones, outside the laboratory that determine the effective spacetime structure inside it. The agent can only indirectly probe this effective spacetime structure through the quantum matter physics in their lab, which is all they can control. In that sense, the effective spacetime structure is a potential source for providing an `external' reference (due to (b) it depends on the outside of the lab) relative to which the agent might want to orient their lab.

We now make a non-trivial requirement, which we re-express in terms of an operational principle shortly: The effective spacetime structure in which the agent and their laboratory reside (and which they cannot control) corresponds to the coarse-grained large-scale limit of a special class of quantum gravity states such that the {\it net interaction} of the matter inside the laboratory with quantum gravitational degrees of freedom is zero on average. That is, we assume any {\it direct} coupling of matter to quantum gravitational degrees of freedom in these states to be washed out through renormalization; quantum gravitational degrees of freedom, other than providing the ambient spatiotemporal structure, have become irrelevant at the laboratory scales in a renormalization group sense. In consequence, the matter, while living in the effective spacetime created by quantum gravitational degrees of freedom, does not further interact with them and can be treated independently for all practical purposes. This is similar to the requirement in Jacobson's derivation of the semiclassical Einstein equations from entanglement equilibrium that there is energy conservation for the large-scale physics and thereby no dissipation or leakage into the ultraviolet physics \cite{Jacobson:2015hqa} (however, here we do not presuppose Lorentzian spacetime structures).

Note that this does {\it not} mean that we assume it to be in principle impossible for the agent to test quantum gravity proper in their lab. What it means is that the matter degrees of freedom in the laboratory, at the relevant laboratory scales, cannot get correlated through anything else than their {\it direct} interactions, which the agent can, in principle, control.\footnote{This includes their direct  (effectively classical) gravitational interaction. The agent will not be able to isolate any quantum subsystem from gravitational interactions with the remaining matter in the laboratory. But the point is that the agent can control the matter systems, move them around and in this sense, while not being able to switch it off, also control the direct gravitational interaction between them. By contrast, the agent is assumed unable to `move quantum gravitational degrees of freedom around'. } That is to say, the matter in the lab, at the scales relevant for our discussion, can{\it not} get correlated through indirect interaction via quantum gravitational degrees of freedom that generate the effective spacetime environment.

This requirement can be regarded as a strengthening of the notion of local inertial frame; the local matter inside the laboratory is not only isolated from the matter outside it, but also from direct interactions with effective quantum gravitational degrees of freedom. It is now a system that is as isolated as it gets while still residing in an effective spacetime structure. It is as if the matter in the laboratory sees an effective vacuum of {\it both} the remaining matter in the universe and quantum gravitational degrees of freedom.

In summary, in these coarse-grained quantum gravity states, {\it a local inertial laboratory must be self-sufficient}: there is {\it no} additional structure that can help an agent in defining the orientation of their frame, other than the quantum matter physics inside it that they can, in principle, directly control. We now write this in the form of an operational principle, which we term a {\it local Mach principle} (LMP) as it requires the local matter physics to generate its own reference:

\begin{LMP} 
	In a local inertial frame, an observer can exclusively use relations among the quantum matter systems that they can directly control in their lab to orient their local frame, but these relations suffice to completely specify their frame.
\end{LMP}

Just like the original GMP, this statement may be technically interpreted in various ways, and we shall specify our technical incarnation of it shortly. However, given that any local frame orientation is exclusively determined with respect to relations among local quantum matter systems, the LMP already entails that the set of transformations between different possible frame orientations must emerge from the structure and relations of matter quantum systems in the local laboratory --- and nothing else. 
We will also discuss the interpretation of possible violations of the LMP in more detail in sec.\ \ref{sec_nolmp}.

In summary, (a) is answered by the EEP, (b) by the GMP and (c) will now be answered by the LMP.

\subsection{A local quantum Mach principle: formalization and implications}\label{sec_lmp}

The local Mach principle posits that an agent can only use relations among the quantum matter systems that they can directly control in their lab to define an orientation of their local reference frame, i.e.\ to ultimately specify, e.g., what they mean by their local $x-, y-,z-$directions and standard clock. Let us now specify in more detail what this part of the LMP means and how it can be formulated in the context of local quantum theory, describing the matter in the local inertial laboratory. First, in order to talk about relations among the quantum matter systems in the laboratory, the agent must be able to subdivide the quantum matter into subsystems and there should be distinguished subsystems that the agent can access separately and relate to one another. For instance, these could be single particles or specific quantum field modes. Second, in line with our assumption that the agent can, in principle, control the matter in their laboratory, we assume them to be able to switch interactions between the matter components on and off, as long as these operations are consistent with the global Hamiltonian $\hat H$. (Recall that we have argued under (a) above that a local inertial frame evolves according to a time-independent Hamiltonian $\hat H$.) Fundamentally, this operational structure comes from the freedom of specifying local initial states (e.g., on the local operator algebra describing the agent's laboratory).

According to the LMP, the {\it only} structure that the agent now has available to define their laboratory's orientation are these quantum structures and, hence, in particular, the observable algebras or Hilbert spaces describing the most fundamental quantum matter constituents into which they are able to divide their lab. 
Defining a frame orientation means defining also a description of the said quantum subsystems and thus to actually choose an operator or Hilbert space basis for them. Now the algebras or Hilbert spaces of the most fundamental matter constituents have natural (e.g., unitary) symmetries. Without any external extra structure, no basis choice among those related by these symmetries will be operationally distinguished. We can thus anticipate that the symmetries of the most fundamental matter constituents will quantify the freedom of the agent in choosing the orientation of their frame and thereby ultimately also lead to transformations among different choices of frame orientation.

The fundamental matter degrees of freedom (e.g., single particles or specific field modes) come with momentum degrees of freedom so that the local quantum theory describing the full content of the laboratory will certainly be infinite-dimensional. However, we now make a  simplifying assumption:
\begin{ass}
For the orientation of the agent's frame, it is sufficient to restrict to \emph{discrete} matter degrees of freedom, i.e.\ to finite-dimensional subsystems of their lab. Hence, we assume that all information about spacetime orientation can be encoded in discrete degrees of freedom that, under suitable operational conditions, the agent can treat as standalone.\footnote{As in the Wigner representations, these discrete degrees of freedom might depend on the momentum mode so that this might require the agent to fix the momentum first.}
\end{ass}
While this is  clearly a convenient simplification for the subsequent discussion, we will conjecture below that this assumption is not actually essential and can ultimately be dropped without modification of the main physical implications. 

It is clear that these discrete degrees of freedom must admit a direct spacetime interpretation if they are to define an orientation in it. For example, they could encompass the spin of massive particles or helicity of photons. But they can{\it not} correspond to the energy levels of an atom. In particular, these discrete matter degrees of freedom can{\it not} be effective degrees of freedom; for instance, quantum dots also define qubit degrees of freedom, but one cannot interpret them in a spacetime sense. Hence, we will take these discrete degrees of freedom to be degrees of freedom of the most fundamental matter constituents into which the agent is able to chop up their lab. We note that in Minkowski spacetime this assumption is  satisfied: e.g., we could, in principle,  define (and also communicate) a frame orientation relative to the helicity of photons.

To build up intuition, we will now explain by means of a toy example how the LMP entails constraints on the Hamiltonian $\hat H$ under all these assumptions.

\begin{example}[Heisenberg model]
	\label{ExHeisenberg}
	Let $\vec S=(S^x,S^y,S^z)$ be spin-$1/2$ angular momentum operators. Consider the Hamiltonian of the length-$n$ spin chain
	\[
	\hat H=-J \sum_{i=1}^{n-1} \vec S_i\cdot \vec S_{i+1}- h\sum_{i=1}^n S_i^z,
	\]
	which we interpret as follows. The constant $J$ describes the interaction between adjacent spins, while $h$ describes the strength of an external magnetic field. In this model, $h$ assumes the role that the effective quantum gravity degrees of freedom take in the discussion of the previous subsection.	The global Hilbert space has well-defined subsystems, corresponding to the single-site Hilbert spaces. If $h\neq 0$, then a hypothetical observer (who is not modelled explicitly, but subject to the assumptions that we have specified above) can determine a distinguished $z$-direction, i.e.\ a distinguished observable $S_i^z$ among all local operators $A_i$. Hence, the observer can use an external reference (the magnetic field, which they cannot control) to partially orient their frame, i.e.\ to partially determine a basis in the space of observables (or equivalently, a Hilbert space basis). This contradicts the colloquial form of the LMP. (Strictly speaking, in this interpretation even the assumption that the laboratory is inertial is violated as the matter inside it has a non-vanishing interaction with the matter of the outside world.)
	
	On the other hand, if $h=0$, then this Hamiltonian conforms with the LMP: any observer can, for example, choose an arbitrary frame on the first site, $i=1$, and use the interactions between the sites to determine a frame (and Hilbert space basis or operator basis) on all other sites. But this is the best they can do; there is no ``absolute direction'' encoded into $\hat H$, which is manifest in the fact that $\hat H$ has ${\rm SU}(2)$ symmetry.
	
	But suppose that we interpret the $n$ spins as describing only \emph{part} of an inertial frame's full quantum system, such that $\hat H$ is an effective Hamiltonian of a subsystem (and the magnetic field is another quantum system in a coherent state within the laboratory), then the LMP would not be violated.

\end{example}

With this motivating example and our finite-dimensionality assumption in mind, we now state a possible formalization of the LMP in the quantum case. To state it, we will make use of the notion of a ``quantum subsystem'', which is a full matrix subalgebra $\mathcal{A}$ of the laboratory's operator algebra. Operationally, we assume that there is a well-defined way of accessing, controlling, and measuring this quantum subsystem, and we can associate a finite-dimensional Hilbert space $\mathcal{H}$ to it such that $\mathcal{A}$ is isomorphic to the set of bounded operators on that Hilbert space. Quantum subsystems $\mathcal{A}_1,\ldots,\mathcal{A}_n$ will be called \emph{disjoint} if $[A_i,A_j]=0$ for all $A_i\in\mathcal{A}_i,A_j\in\mathcal{A}_j,i\neq j$. 

We need one more ingredient to formulate the LMP. What we would like to state, among other things, is that there are quantum subsystems (e.g.\ a single site of the Heisenberg spin chain in Example~\ref{ExHeisenberg}) that have no distinguished frames whatsoever. In quantum theory, a frame will be a basis of observables, an \emph{operator basis} (or, equivalently, a basis in Hilbert space). An operator basis for a $d\times d$ matrix algebra $\mathcal{A}$ is a set of $d^2$ matrices in $\mathcal{A}$ that are linearly independent.  For example, the three Pauli matrices together with the identity constitute an operator basis in the case $d=2$. Note, however, that the latter claim is an abstract mathematical statement; for any physical subalgebra $\mathcal{A}$, a choice of basis corresponds to a choice of how to encode the physical observables into matrices. Now, claiming that literally \emph{all} frames (operator bases) are equivalent would clearly be wrong: for example, in many cases, observers can use their physical tools to set up an \emph{orthonormal} Hilbert space basis, and those would be distinguished from other bases which are not orthonormal. This motivates to introduce the following definition:

\textbf{Equivalent operator bases.} \emph{Suppose we fix a set of physical background assumptions, specifying what observers can always operationally accomplish in their lab. Then two operator bases (or, more generally, two sets\footnote{The sets may have additional structure, e.g.\ they may be ordered (as in operator bases) or have a topology (if they are open subsets).} of operators) $S$ and $S'$ will be called \emph{equivalent} {if these background assumptions} \emph{alone} do not let the observer distinguish $S$ and $S'$ operationally.}

For example, we will fix a set of background assumptions in Observation~\ref{ObsGop} below: we will say that observers can use their measurement devices to unequivocally determine the \emph{eigenvalues} of physical observables (in addition to their linear structure). Now define $S$ to be a finite set of operators that all have, for example, negative eigenvalues, and $S'$ a set of the same size with operators that all have positive eigenvalues. Then $S$ and $S'$ can be distinguished by the observer, and they are not equivalent. On the other hand, if $S$ and $S'$ are two sets of operators such that $U S U^\dagger = S'$ for some unitary $U$, then they are equivalent. In Example~\ref{ExHeisenberg}, the single-site observable $\{S_i^z\}$ would be equivalent in this sense to $\{S_i^x\}$, but both observables could be operationally distinguished if $h\neq 0$ due to the form of the Hamiltonian. This is a violation of what will be condition (i) in the following technical formulation of the LMP:
\begin{LMPt}
	For every local inertial frame, there is a finite set of disjoint finite-dimensional quantum subsystems $\mathcal{A}_1,\ldots,\mathcal{A}_n$ such that the following holds:
	\begin{itemize}
		\item[(i)] If $S$ and $S'$ are two equivalent subsets of operators on any $\mathcal{A}_i$ (for example, two equivalent operator bases), then there is \emph{no way} to distinguish $S$ operationally from $S'$.
		\item[(ii)] Specifying operator bases for all $\mathcal{A}_i$ determines operationally uniquely a choice of {operator basis} for {the full laboratory.}
	\end{itemize}
	We call the algebras $\mathcal{A}_1,\ldots,\mathcal{A}_n$ \emph{parent subsystems}, and the algebra $\mathcal{A}_1\otimes\ldots\otimes\mathcal{A}_n$ \emph{parent subalgebra}.
\end{LMPt}
In other words, item (i) posits that there is no extra physical structure that breaks the symmetry implied by the physical background assumptions. In the following, we will make two different kinds of physical background assumptions. This will lead to two different notions of ``equivalent operator bases'' and thus to different consequences of the LMP which we describe in Observations~\ref{ObsGop} and~\ref{ObsGop2}.

For what follows, we conjecture that we do not need to assume that the $\mathcal{A}_i$ are finite-dimensional; simply assuming that they are von Neumann algebra factors should be sufficient. However, to avoid technicalities, we will henceforth work with the assumption of finite-dimensionality. It is clear that every $\mathcal{A}_i$ must live on a Hilbert space that has dimension at least two (i.e.\ is at least a qubit), since otherwise the notion of operator basis of $\mathcal{A}_i$ would be trivial and the corresponding system could be disregarded.

Let us briefly illustrate the technical form of the LMP by example of the Heisenberg model, Example~\ref{ExHeisenberg}. In this case, we can choose any of the sites (for example the first one) to play the role of the parent subsystem $\mathcal{A}_1$, i.e.\ $\mathcal{A}_1$ will be the $2\times 2$ matrix subalgebra of observables on the first spin. We have $n=1$, i.e.\ there are no other distinguished subsystems that we need to pick to satisfy the LMP. Choosing an operator basis on the first spin uniquely determines operator bases on all others, as explained above. Hence, item (ii) of the LMP is satisfied. Furthermore, if $h=0$ then also item (i) is satisfied, i.e.\ the global ${\rm SU}(2)$ symmetry of the model prevents one from distinguishing any two unitarily equivalent subsets of operators on $\mathcal{A}_1$.

To see that sometimes we can have $n>1$, consider two independent, non-interacting Heisenberg spin chains. Then we have $n=2$, and $\mathcal{A}_i$ corresponds to a single-site operator algebra on the $i$th spin chain.

In comparison to the colloquial form of the LMP, item (i) expresses the fact that there is no external reference with which to orient a frame, and item (ii) expresses sufficiency: choosing frames on the parent subsystems uniquely defines a frame on the full laboratory. As Example~\ref{ExHeisenberg} illustrates, we can think of the origin of this sufficiency as ultimately coming from the fundamental interaction between the laboratory's subsystems. Without interaction, subsystems could not be compared or related. Parent subsystems must thus have the property that they interact, directly or indirectly, with all other quantum subsystems of the inertial frame.

As Example~\ref{ExHeisenberg} furthermore demonstrates, validity of the LMP implies that the system has a fundamental symmetry. To determine this symmetry in general, we have to be a bit more specific about the physical background assumptions that allow an observer to pick an operator basis, and thus to encode an observable $\hat O\in\mathcal{A}$ into some matrix, say, $\varphi(\hat O)$. As explained above, the LMP and its consequences will depend on those background assumptions, and for now these will be the following: We assume that both the notion of self-adjointness as well as the notion of linearity on the observables are physically evident --- in other words, given two physical observables $\hat O_1,\hat O_2$ and scalars $\lambda,\mu$ together with a third $\hat O$, quantum theory itself predicts that every observer will agree on whether $\hat O=\lambda\hat O_1+\mu\hat O_2$ or not. Moreover, we assume that the set of \emph{eigenvalues} of any observables $\hat O$ is directly physically accessible --- that is, the set of possible outcomes of measuring $\hat O$. Then any two mathematical descriptions $\varphi(\hat O)$ and $\varphi'(\hat O)$ of an observable $\hat O$ are related by~\cite{Jafarian}
\[
\mbox{either }\varphi'(\hat O)=U\varphi(\hat O) U^\dagger \mbox{ or }\varphi'(\hat O)=U\varphi(\hat O)^\top U^\dagger,
\]
where $U$ is some unitary.
It depends on the physical background assumptions whether we consider the transpose, $M\to M^\top$, a symmetry or not. If the background physical assumptions allow an observer to determine whether a given observable $\hat O$ is the \emph{product} of two other given observables $\hat O_1$ and $\hat O_2$, then any $\varphi$ can be chosen as an algebra homomorphism, and the transpose map is ruled out. (In other words, physics will in this case determine the parity of the observable encoding.) We will now make this assumption for simplicity, but think that it would not substantially change our conclusions if one included the transpose as a possible symmetry.

Since we can have observables that live on the tensor product of the $\mathcal{A}_i$ (or, equivalently, entangled states across the corresponding composite quantum system), and since product bases are canonical operator bases, encodings of elements of all parent algebras are related by maps of the form $U_1\otimes\ldots\otimes U_n \bullet U_1^\dagger\otimes\ldots\otimes U_n^\dagger$, with all $U_i$ unitary. Without loss of generality, we may assume that every $U_i$ has unit determinant.
\begin{observation}
	\label{ObsGop}
	If an inertial frame satisfies the conditions of the LMP, then observers always have some remaining fundamental freedom of choice of operator basis. Under the background assumptions that we have just stated, the different choices of basis are related by a fundamental symmetry group $\mathcal{G}^{\rm op}$, which is
		\[
		\mathcal{G}^{\rm op}=\left\{A \mapsto 
		U_1\otimes U_2\otimes\ldots\otimes U_n \,A \,U_1^\dagger\otimes U_2^\dagger\otimes\ldots\otimes U_n^\dagger\,\,|\,\, U_i\in{\rm SU}(n_i,\mathbb{C})
		\right\},
		\]
		where $n_i$ denotes the Hilbert space dimension of $\mathcal{A}_i$.
	In particular, since $n\geq 1$ and every $\mathcal{A}_i$ is supported on a Hilbert space of dimension $n_i\geq 2$, and since ${\rm PSU}(2)\simeq {\rm SO}(3)$, we have ${\rm SO}(3)\subseteq \mathcal{G}^{\rm op}$.
\end{observation}
Thus far, we have assumed that the eigenvalues of any physical observable $\hat O$ are unambiguously determined by the physics, i.e.\ independent of the choice of reference frame. Intuitively, this corresponds to situations where the eigenvalues are themselves abstract, ``speakable''~\cite{MarvianSpekkens,Bartlett:2007zz} real numbers (e.g.\ probabilities as for density operators, or multiples of $\hbar/2$ for angular momentum operators), but not ``unspeakable'' physical quantities themselves (e.g.\ spatial distances), measured in arbitrary physical units. For such observables, the ``unspeakable'' eigenvalues of the mathematical descriptions $\varphi(\hat O)$ will in general depend on the choice of reference frame. The simplest instance of this is the case when the eigenvalues can be measured either in some unit $u$ or another unit $u'$ (think of miles versus meters, for example); then there is some factor $r>0$ such that $\lambda_i(\varphi'(\hat O))=r\, \lambda_i(\varphi(\hat O))$, where $\lambda_i$ denotes the $i$-th eigenvalue, namely the $r$ such that $u=r u'$.

It may at first seem as if this scaling of eigenvalues was the only freedom in describing the observable. However, we will now argue that our setup (local quantum physics in some spacetime) motivates situations with even more symmetry, and we will demonstrate in Subsection~\ref{SubsecExMinkowski} below that this kind of symmetry is in fact realized in ordinary Minkowski space: it describes measurements for which the finitely many outcomes correspond to the different possible deflections of a wave packet, as in a Stern-Gerlach device.

How would different observers (with different frames of reference) describe the eigenvalues of such a measurement? In such a setup, an eigenvalue $\lambda=0$ has a clear and direct operational interpretation: it is an outcome for which \emph{the wave packet is not deflected at all}. Whether there is any other extra structure that would single out natural descriptions of the eigenvalues depends on the physical background assumptions. We will now make the (arguably strong, but not implausible) assumption that \emph{there are no such extra structures}. That is, the only structure on the physical observables of $\mathcal{A}$ that different observers can unequivocally agree on is
\begin{itemize}
	\item[(i)] \emph{the notion of self-adjointness}: the subspace $\mathcal{A}^{\rm sa}$ of self-adjoint elements of $\mathcal{A}$ is physically distinguished, because it corresponds to those observables that can actually be measured.
	\item[(ii)] \emph{The linearity structure on $\mathcal{A}$}: it is objectively clear how to build linear combinations of physical observables. Due to quantum mechanics, this linear structure is fundamental, since the whole point of observables is to allow to assign expectation values to states (and vice versa)~\cite{BratteliRobinson}.
	\item[(iii)] \emph{The number of zero eigenvalues}: it is objectively clear whether a quantum system in the device experiences a change of state (deflection or acceleration). This means that if $\varphi$ and $\varphi'$ denote different ways to encode physical observables into mathematical descriptions, then ${\rm rank}\,\varphi(\hat O)={\rm rank}\, \varphi'(\hat O)$.
	
	This is clearly a semiclassical approximation: in general, determining whether a wavepacket has shifted or not cannot be done with unit probability. Nevertheless, we are working in a regime in which one can reliably distinguish the $n$ outcomes by observing the wavepacket, and in this case, this approximation is justified. 
\end{itemize}
Thus, $\varphi(\hat O)$ and $\varphi'(\hat O)$ are related by an invertible linear rank-preserving map on the Hermitian matrices, and this implies~\cite{Loewy} that there is some $\varepsilon\in\{-1,+1\}$ and an invertible matrix $X$ such that
\[
\mbox{either }\varphi'(\hat O)=\varepsilon\, X\varphi(\hat O)X^\dagger \mbox{ or }\varphi'(\hat O)=\varepsilon\, X\varphi(\hat O)^\top X^\dagger.
\]
At first sight it may seem puzzling that two descriptions of the same observable are related by conjugation with a map $X$ that is \emph{not unitary.} If we apply this map to states, doesn't this mean that the total probability is not preserved? The puzzle can be resolved by acknowledging that $X$ is in general a map from a Hilbert space $\mathcal{H}$ to \emph{another} Hilbert space $\mathcal{H}'$. As linear spaces, both are equal to $\mathbb{C}^n$, but they carry different inner products. Consequently, $X$ is an isometry, which is allowed by the textbook axioms of quantum mechanics. If $\mathcal{H}=\mathcal{H}'$ then $X$ a unitary matrix. This will be explained in more formal detail for the concrete case of spin-$1/2$ particles in Minkowski space in Subsection~\ref{SubsecExMinkowski} below.

We can now reconsider the LMP under these modified physical background assumptions: Hilbert-Schmidt orthogonality will not be of use any more to single out distinguished descriptions of observables, and neither will the eigenvalues (except for the zero eigenvalues). Thus, Observation~\ref{ObsGop} becomes modified and gets replaced by the following:
\begin{observation}
	\label{ObsGop2}
	If an inertial frame satisfies the conditions of the LMP \emph{under the modified physical background assumptions {(i), (ii) and (iii)} above}, then the observer's choices of operator basis are related by the fundamental symmetry group
	\[
	\mathcal{G}^{\rm op}\simeq \left\{\left.
	A\mapsto \lambda(Y_1\otimes\ldots\otimes Y_n)\tau(A)(Y_1^\dagger\otimes\ldots\otimes Y_n^\dagger)\,\,\right|\,\, Y_i\in {\rm SL}(n_i,\mathbb{C}),\, \tau(A)\in\{A,A^\top \},\,\, \lambda\in\mathbb{R}\setminus\{0\}
	\right\},
	\]
	where $n_i$ denotes the Hilbert space dimension of $\mathcal{A}_i$. In particular, since $n\geq 1$ and every $\mathcal{A}_i$ is supported on a Hilbert space of dimension at least two, $\mathcal{G}^{\rm op}$ contains at least the subgroup generated by the maps $A\mapsto \pm r Y A Y^\dagger$ and $A\mapsto \pm r Y A^\top Y^\dagger$, where $r>0$ and $Y\in{\rm SL}(2,\mathbb{C})$. Note that the group of maps $A\mapsto YAY^\dagger$ is ${\rm PSL}(2,\mathbb{C})$, and ${\rm PSL}(2,\mathbb{C})\simeq {\rm SO}^+(3,1)$ for the proper orthochronous Lorentz group. The overall sign and the transpose generate space and time inversions, extending this group to the full Lorentz group ${\rm O}(3,1)$. In summary, we have $\mathbb{R}_+\times {\rm O}(3,1)\subseteq \mathcal{G}^{\rm op}$.
\end{observation}
{To see this, note again that a natural choice of operator basis on the product of the $\mathcal{A}_i$ is given by a product of operator basis on each algebra. For every $\mathcal{A}_i$, all such choices are related by either $A_i\mapsto \varepsilon_i X_i A_i X_i^\dagger$ or $A_i\mapsto \varepsilon_i X_i A_i^\top X_i^\dagger$, where $\varepsilon_i\in\{-1,+1\}$ and $\det X_i\neq 0$ (here $i$ labels subsystems, not basis elements). Since the full product transformation must itself satisfy (i), (ii) and (iii), we cannot have transpositions on only some of the $\mathcal{A}_i$, but we must have either $\varphi'(\hat O)=\varepsilon(X_1\otimes\ldots\otimes X_n)\varphi(\hat O)(X_1^\dagger\otimes\ldots\otimes X_n^\dagger)$ or $\varphi'(\hat O)=\varepsilon(X_1\otimes\ldots\otimes X_n)\varphi(\hat O)^\top(X_1^\dagger\otimes\ldots\otimes X_n^\dagger)$. In the case without the transpose, define $r_i>0$ and $\theta_i\in[0,2\pi)$ via $\det X_i=r_i e^{i\theta_i}$, and set $Y_i:=r_i^{-1/d_i} e^{-i \theta_i/d_i}X_i$, where $X_i$ are $d_i\times d_i$ matrices. Then $\det Y_i=1$ and $\varphi'(\hat O)=\varepsilon r(Y_1\otimes\ldots\otimes Y_n)\varphi(\hat O)(Y_1^\dagger\otimes\ldots\otimes Y_n^\dagger)$, with $r=r_1^{2/d_1}\cdot\ldots\cdot r_n^{2/d_n}>0$, $\varepsilon\in\{-1,+1\}$ and $Y_i\in{\rm SL}(\mathcal{A}_i)$. The case with the transposition is analogous.
	
	The physical background assumptions (i)---(iii), and also the result of Observation~\ref{ObsGop2}, are in fact physically realized by spin observables in Minkowski space. This is explained in detail in the next subsection.

\subsection{Example: spin-$1/2$ particles in Minkowski space}
\label{SubsecExMinkowski}
\textit{\textbf{Hint for impatient readers:} This subsection may be skipped on a first reading --- its content is not essential for the main flow of arguments.} However, it demonstrates the physical significance of Observations~\ref{ObsGop} and \ref{ObsGop2}.\\

Let us see how the technical form of the LMP, under the two different sets of physical background assumptions, is realized within relativistic quantum mechanics in Minkowski space.

\textbf{a)} We start with the \textbf{physical background assumptions that lead to Observation~\ref{ObsGop}}: observers will in general disagree on a basis in Hilbert space (if they have not set up a joint reference frame beforehand), but they will \emph{agree} on the eigenvalues of observables. In the case of relativistic quantum mechanics, subsystems of the laboratory Hilbert space will correspond to particles with wavefunctions
\[
   \psi(\vec p, t,\sigma),
\]
where $\vec p$ is the momentum, $t$ is time, and $\sigma=-s,-s+1,\ldots,s-1,s$, where $s$ is the spin quantum number. Specifying the technical form of the LMP, it turns out that we have $n=1$, i.e.\ a single parent subsystem $\mathcal{A}_1$. Namely, consider the spin degree of freedom of an arbitrary massive single spin-$1/2$ particle, such as an electron, and let $\mathcal{A}_1$ be the corresponding observable algebra (which is isomorphic to the $2\times 2$ matrices). We assume that we can operationally access these observables, for example by approximating the particle's state of motion by a momentum eigenstate, and by transforming into the reference frame in which the particle is at rest. Operationally, this amounts to shooting the particle into a Stern-Gerlach device, and measuring $\mathcal{A}_1$ observables by observing the deflection of its wave packet in a magnetic field that is set externally into a specific direction of inhomogeneity. The two possible deflections are interpreted as yielding outcome eigenvalues $\pm \hbar/2$, and all observers agree on these two values.

A choice of operator basis amounts to specifying a set of observables $X,Y,Z$, i.e.\ of the three Pauli matrices -- which is equivalent to specifying a spatial frame of reference, defining three spatial directions $x,y,z$. This obviously implements the ${\rm SO}(3)$-symmetry that we have found in Theorem~\ref{thm_gopgor}, and does so by conjugation with ${\rm SU}(2)$-elements, $U\bullet U^\dagger$. Operationally, once we have specified $x,y,z$-directions, this gives us operator bases for the internal degrees of freedom of all other spin-$s$ particles -- since it gives the observable $S_{\vec n}$, the ``spin in direction $\vec n$'', an unequivocal meaning. This specifies a corresponding operator basis since the $S_{\vec n}$ (and their eigenprojectors) are tomographically complete on the corresponding finite-dimensional Hilbert spaces; for some more details on this, see \cite{Hoehn:2014vua}.

Similarly, observers that agree on $x,y,z$ directions will then agree on eigenstates of the momentum operator like $|p_x\rangle$. This is because they will all agree on when a particle moves in $x$-direction, and our background assumptions also tell us that they agree on the eigenvalues of the momentum operator -- that is, on the unit with which they measure momentum. Therefore, if only they agree on the operator basis for a single spin-$1/2$ particle, they automatically agree on operator bases for the complete Hilbert spaces of all particles, and thus of all of their laboratory.

Geometrically, we can interpret this situation as saying that the reference frames of all observers are at rest with respect to each other (and with respect to the Stern-Gerlach devices that are used to implement the $\mathcal{A}_1$-measurements) --- or, equivalently, that they all have agreed on a common time variable $t$. This breaks the Lorentz symmetry of Minkowski space down to rotational symmetry.\\

\textbf{b)} What changes if we turn to our \textbf{second set of physical background assumptions, (i) -- (iii), leading to Observation~\ref{ObsGop2}}? This describes a situation where different observers may disagree on the eigenvalues of observables, and the easiest way to imagine such a situation is by having eigenvalues which themselves carry a spacetime interpretation. Reconsider the observables $\mathcal{A}_1$ of the internal degree of freedom of a spin-$1/2$ particle. But now, let us regard the two outcomes of a Stern-Gerlach device as encoding additional information rather than ``spin up'' or ``down'': namely, we scale the spin observables such that, in the particle's rest frame, the two outcomes carry eigenvalues $\pm |\vec G|$, where $\vec G\in\mathbb{R}^3$ denotes the particle's acceleration due to the inhomogeneity of the magnetic field. For one of the spin eigenstates, the particle will experience an acceleration parallel to the direction of inhomogeneity, and for the other eigenstate it will be antiparallel. The eigenvalues therefore encode additional geometric physical information on the post-measurement state of the particle.

For two observers whose frames of reference are rotated with respect to each other, the accelerations $\vec G$ and $\vec G'$ that they assign will be related by the corresponding rotation $R$, i.e.\ $\vec G'=R\vec G$, and hence their descriptions of the corresponding observables $\hat S$ and $\hat S'$ will differ by a unitary, $\hat S'=U_R \hat S U_R^\dagger$, where $U_R$ is the spin-$1/2$ representation of the rotation $R$. In particular, both observers will agree on the eigenvalues of the spin observable (for both these are $\pm\hbar/2$).

But now a new possibility arises: namely, that one observer is not only rotated, but \emph{boosted} with respect to the other. What happens in such a case? If $\Lambda$ is the Lorentz transformation that transform from the particle's rest frame the the other observer's rest frame, then the purely spatial acceleration 4-vector $G=(0,\vec G)$ will be transformed into $G'=(G'_0,\vec G')$, where $G'_\nu = \Lambda_\nu^{\enspace \mu}G_\mu$.

Physically, the Stern-Gerlach experiment will still constitute a valid experiment for the boosted observer, one with two possible outcomes in a unit of acceleration. Which observable $\hat S'$ will this observer use to describe the measurement? There is only one possible choice that amounts to a meaningful transformation rule, namely that $\hat S'=X\hat S X^\dagger$, where $X$ is the corresponding ${\rm SL}(2,\mathbb{C})$--representation\footnote{Whether this is a left-handed or right-handed spinor is irrelevant for this argument.} of the boost $\Lambda$. Since $X$ is not unitary, this transformation will change the eigenvalues of the spin observable. A simple calculation shows that the old eigenvalues $\lambda_\pm = \pm |\vec G|$ change to $\lambda'_\pm=G'_0\pm|\vec G'|$. That is, the boosted observer will in general see the two wave packets (corresponding to the two spin eigenstates) to be accelerated \emph{asymmetrically}; this is explained in more detail in \cite{Hoehn:2014vua}. All this is compatible with the prediction of Theorem~\ref{thm_gopgor}, namely, that we obtain an ${\rm SO}^+(3,1)$-symmetry which is represented by ${\rm SL}(2,\mathbb{C})$ conjugation.

In this case, a choice of operator basis for the spin-$1/2$ internal degree of freedom $\mathcal{A}_1$ amounts to fixing not only a spatial, but a \emph{spacetime} frame of reference $(x,y,z,t)$. Similarly as in part a) of this example, this determines operator bases for all degrees of freedom of the laboratory.\\

How can a transformation $\hat S\mapsto X\hat S X^\dagger$ be quantum-mechanically meaningful if $X\in{\rm SL}(2,\mathbb{C})$ is not a unitary matrix? There are three complementary ways to see how this fits into the standard quantum formalism. We can consider the Hilbert space $\mathbb{C}^2$ of $\mathcal{A}_1$ as implicitly carrying structure that depends on the observer's particular 4-momentum $p$ (therefore being denoted $\mathcal{H}_p$), and $X$ as a map from $\mathcal{H}_p$ to $\mathcal{H}_{\Lambda p}$, where $\Lambda$ is the Lorentz transformation that corresponds to $X$. If $p=\Lambda p$, then $X$ will be a unitary matrix that represents the corresponding rotation from the little group, preserving the Hilbert space. But if $\Lambda p\neq p$, then $X$ will formally be an \emph{isometry}, mapping from one Hilbert space to another. Both $\mathcal{H}_p$ and $\mathcal{H}_{\Lambda p}$ are $\mathbb{C}^2$ as vector spaces, but they carry different inner products such that the resulting map $X$ is an isometry. (The resulting inner product on $\mathcal{H}_{\Lambda p}$ will be given by $\langle\psi|(XX^\dagger)^{-1}|\varphi\rangle$; for more details on this, see \cite{Hoehn:2014vua}). This formalism has first been derived (though in different notation) via a WKB approximation in \cite{Palmer:2011bt,Palmer:2013lva}.

A complementary but equivalent point of view acknowledges that the elements of $\mathcal{A}_1$ are not ordinary scalar matrices, but \emph{matrices with physical entries that carry the unit of acceleration}, $m/s^2$. But then, if we take $A,B\in\mathcal{A}_1$, their matrix product $AB$ carries the unit $m^2/s^4$ and is not an element of $\mathcal{A}_1$ any more. That is, $\mathcal{A}_1$ is not an operator algebra in the strict sense, and we have to explicitly define an additional multiplication structure on it. (It does, however, carry a natural linear structure, a notion of self-adjointness, a notion of rank, and a cone of positive semidefinite elements.) Now, choosing a normalization in Hilbert space is equivalent to specifying the pure states as rank-one projectors $P$ in the space of observables, i.e.\ the self-adjoint rank-one elements with $P^2=P$. To do so, we have to define a multiplication on the linear space $\mathcal{A}_1$. This can be done by demanding that for $p=(p_\mu)=(E/c,0,0,0)$, the multiplication in $\mathcal{A}_1$ is simply matrix multiplication, while for $p':=\Lambda p$ it is given by $A\cdot B:=A(XX^\dagger)^{-1}B$, where $X\bullet X^\dagger$ represents $\Lambda$ via ${\rm SL}(2,\mathbb{C})$-conjugation. Again, $\mathcal{A}_1$ carries momentum-dependent extra structure that leads to consistency with textbook quantum mechanics.

Third, note that we are considering the transformation of a description of a \emph{discrete} subsystem of the local laboratory. While the specific description of this transformation, when writing down concrete matrices, can look like a non-unitary matrix, its impact on the \emph{total} operator algebra of the full laboratory can (and will) still be unitary. By considering the total lab, it will turn out that the transformations like $X\bullet X^\dagger$ map (as just explained) the subsystem effectively to another subsystem, which corresponds to an isometry, compatible with global unitarity. {Note that the eigenvalues (which can change under such transformations) are the result of a semiclassical approximation, as explained in (iii) above. Indeed, this comes back to the finite-dimensionality assumption at the beginning of this subsection, where we assumed the agent to be able to isolate, \emph{under suitable operational conditions}, the discrete degrees of freedom from, e.g.\ the momentum modes. This separation might depend on the momentum and here we effectively assume the momentum states to be semiclassical. The transformation $X\bullet X^\dagger$ would then effectively map a discrete subsystem with fixed semiclassical momentum to another discrete subsystem with semiclassical momentum. This underscores the non-fundamental nature of non-unitarity in this formalism. 

In this example, we have worked in the setting of relativistic quantum mechanics, and it would be interesting to generalize these considerations to quantum field theory. We conjecture that the main insights of our analysis remain valid, but that interesting new aspects come into play in particular due to the existence of gauge fields. We leave this extension to future work.

\subsection{Spacetime frame transformations act on state spaces and observables}\label{sec_wald}

It is a priori not clear what the relation of $\cg^{\rm op}$ is to geometric frame transformations that come from the spacetime structure itself. It is, however, clear that there should be {\it some} relation because the local frame transformations originating in the spacetime structure itself precisely translate from one local description of physics to another and this is also what $\cg^{\rm op}$ does. We will now make this relation step by step more precise.

For a moment we reverse the perspective and focus on local spacetime transformations and how they act on quantum systems. Indeed, a general argument by Wald \cite{wald} implies that the isometry group of a spacetime should have a natural action on the states of a theory defined on that spacetime. Given that we neither intend to restrict to spacetimes with symmetry nor to metric spacetimes and, in any case, focus here on relations of frame orientations in a local laboratory, we shall now adapt this argument to local {\it internal frame orientation preserving transformations} of generalized geometries, which we define shortly. 

Let spacetime be given by a manifold $\cm$ and a spacetime structure $H$ defined on it, $(\cm,H)$. The spacetime structure $H$ need not be a metric, nor any other structure canonically defining all distances within $\cm$, which is why for now we do not call it a geometry. Instead, we  only require it to define the minimal structure necessary for referring to observer frames and their relations. We will specify this minimal structure in the course of the discussion below. Later, we will restrict $H$ to be a dispersion relation, i.e.\ essentially a Hamiltonian function, which is why we denote it here with an $H$. Hence, $(\cm,H)$ will be a spacetime with a dispersion relation defining the spacetime structure. Note that, for example, in Minkowski space, the dispersion relation is indeed equivalent to the metric. The language of dispersion relations will be more convenient for our operational purposes.

Consider now a family of observers at an event $x\in\cm$. Geometrically, each observer is characterized by a frame, i.e.\ a set of vectors $\{e_\ca\}$ in $T_x\cm$, where $\ca=0,1,2,3$ denotes frame indices. (Equivalently, the observer is characterized by its co-frame in $T_x^*\cm$, which will later be more convenient when working with dispersion relations.) We will not require  $H$ to provide orthonormality conditions for these frames since this is operationally not necessary; our only conditions shall be (i) that $e_0$ is tangent to the worldline of the observer, indicating their direction of time, and (ii) that $e_0$, through $H$, also defines what `spatial' directions for this observer are and that the $e_A$, $A=1,2,3$, span this set. (Capital latin letters will thus be used to indicate spatial frame vectors.) For the moment, we will not specify further what exactly `spatial' means, but, in some cases, we may think of spatial directions as being tangent to initial data surfaces.\footnote{E.g., if $H$ is a dispersion relation, this is the case if it is hyperbolic \cite{Raetzel:2010je}, but not in general.}  The $e_A$ describe the observer's and, in particular, their measurement apparatuses' spatial orientation. Just like in 
special and general relativity, we also assume that two frames $\{e_\ca\}$ and $\{e'_\cb\}$ with non-aligned worldline tangent vectors $e_0\nsim e'_0$ do {\it not} see the same space so that $\{e'_B\}$ does not lie in the span of $\{e_A\}$.

Suppose now there is a physical system with state space  
$\Sigma$ in an infinitesimal neighbourhood $U_x$ of $x$ (e.g., the local laboratory or a subsystem of it). For instance, this could be a local quantum system, described by an observable algebra $\ca_{U_x}$, ascribed to $U_x$, with $\Sigma$ a Hilbert space to represent $\ca_{U_x}$ as in sec.\ \ref{sec_lmp}. In line with our assumption of finite-dimensional quantum systems in sec.\ \ref{sec_lmp}, we shall assume that any state $\sigma\in\Sigma$ is uniquely characterized by $k$ real numbers which correspond to possible outcomes of a complete set of measurements (for quantum systems, these numbers could also be probabilities). That is, for every observer frame $\{e_\ca\}$, we obtain a map $f_e:\Sigma\rightarrow\mathbb{R}^k$ that assigns to each state $\sigma$ the $k$ outcomes of a choice of a complete set of measurements whose apparatuses are oriented according to the $\{e_A\}$.  Hence, a different observer $\{e'_\ca\}$ will define a different map $f_{e'}:\Sigma\rightarrow\mathbb{R}^k$, as measurement outcomes may depend on the state of motion and the orientation of apparatuses. Accordingly, $f_e(\sigma)\neq f_{e'}(\sigma)$ is possible, depending on the physical situation. Notice that these maps are actually associated with the neighbourhood $U_x$, however, for notational simplicity we drop a reference to it.

Given different frames, it will be crucial to consider how the descriptions of the physics with respect to different frame choices may be related. To this end, we require the spacetime structure $H$ to provide sufficient structure for an observer to speak about an operationally meaningful orientation of their frame $\{e_\ca\}$. For the moment, we will not specify the exact meaning of such a frame's orientation further and just assume $H$ is capable of defining it. 
Consider two frames $\{e_\ca(x)\},\{e'_\ca(x')\}$ for now at two distinct events $x,x'\in\cm$. We will say that they have {\it the same operational orientation} if they are oriented identically {relative} to all operationally accessible structures in $U_x$ and $U_{x'}$, respectively, which are ``speakable''. That is, they have the same operational orientation if they are oriented identically with respect to all structures that observers, firstly, have operationally access to in their local laboratories and, secondly, could communicate to one another through classical communication. For example, in empty Minkowski spacetime where $H$ would be (equivalent to) the Minkowski metric, the only ``speakable'' operationally accessible structure would be the local light cone and the {\it relative} orientation and length of the local frame vectors, but not their actual length because the choice of units cannot be communicated without reference to any shared physical system. Hence, the most different observers in (closed laboratories in) empty Minkowski spacetime could agree upon by classical communication is that, up to a choice of units, they locally employ orthonormal frames. As we will see, in spacetimes with dispersion relation, more general structures will be permitted. The reason we make this assumption on $H$ is so that we can meaningfully state whether the orientation of the measurement apparatuses in frame $\{e'_\ca\}$ relative to all ``speakable'' local structure is the same as of those in frame $\{e_\ca\}$. We shall exploit this to formulate a local covariance of physics momentarily.

To this end, consider the set of all frames at $x\in\cm$ that have the same operational orientation as $\{e_\ca\}$. This will be a subset of the space of all frames at $x$\footnote{This is the space of all ordered bases of $T_x\cm$ and thus not a vector space but a subset of four copies of $T_x\cm$.} and defines an equivalence class of local frames. Denote by ${\Phi}^{\rm or}$ the set of all transformations on the space of frames at $x$ that leaves this equivalence class invariant. That is, $\{e_\ca\}$ and $\{\phi(e_\ca)\}$ have the same operational orientation if $\phi\in\Phi^{\rm or}$. As such, $\Phi^{\rm or}$ will be a group and we shall refer to it as the {\it group of internal frame orientation preserving transformations} since the operational orientation does not require external structures and is thereby truly an {\it internal} orientation of the laboratory. 
In principle, $\Phi^{\rm or}$ could depend on the equivalence class of operational frame orientations. However, we assume this not to be the case and will also find this assumption to be satisfied in the spacetimes with dispersion relation considered in sec.\ \ref{sec_finsler}. Clearly, in Lorentzian spacetimes $\Phi^{\rm or}$ would be the Lorentz group.

We now resort to this structure to require a {\it covariance of the local physics}: Every physically possible result of a set of measurements performed by $\{e_\ca\}$ shall also be a physically possible result of a set of measurements conducted by $\{e'_\ca\}$, {\it provided} $\{e'_\ca\}=\{\phi(e_\ca)\}$ for $\phi\in\Phi^{\rm or}$ so that the orientation of each of the complete sets of measurement apparatuses is the same in both frames with respect to all ``speakable'' structure. That is, these two frame choices truly cannot be distinguished using ``speakable'' information only. 
Then clearly there must exist a $\sigma'\in\Sigma$ so that
$f_e(\sigma)= f_{e'}(\sigma')=f_{\phi(e)}(\sigma')$. Hence, every such orientation preserving  $\phi$ defines a map $\tilde \phi:\Sigma\rightarrow\Sigma$, so that frame $\{e_\ca\}$ describes the state $\sigma$ in the same way in which frame $\{e'_\ca\}$ describes $\tilde \phi(\sigma)$. The orientation preserving frame transformations $\Phi^{\rm or}$ thus have an action on the state space $\Sigma$. 

Finally, denote the abstract group, which is isomorphic to the frame orientation preserving transformations $\Phi^{\rm or}$, by $\cg^{\rm or}$ and by $\phi_g\in\Phi^{\rm or}$ the frame transformation corresponding to $g\in\cg^{\rm or}$. We write $\tilde\phi_g:\Sigma\rightarrow\Sigma$ for the corresponding map on state spaces and our argument above implies that we have a group homomorphism
\ba
\tilde\phi_{g_1}\circ\tilde\phi_{g_2}=\tilde\phi_{g_1g_2},
\ea
in analogy to the discussion in \cite{wald}.

Thus far, and for later purpose, we have been very general and focused on the action of $\cg^{\rm or}$ on state spaces. Let us now connect this discussion with the observable algebras of quantum systems in sec.\ \ref{sec_lmp}. From the above discussion it is evident that, conversely, $\cg^{\rm or}$ also has an action on the complete sets of measurements, described by the maps $f_e$. In particular, for some quantum subsystem of the laboratory in sec.\ \ref{sec_lmp}, this map will be associated with a basis of the associated set of self-adjoint observables $\ca^{\rm sa}$. This basis will correspond to the complete set of measurement devices that are oriented according to the frame vectors $\{e_\cb\}$ above. As such, it is clear that $\cg^{\rm or}$ will have some action on $\ca^{\rm sa}$ because after the frame transformation, the measurement devices will be oriented differently and thereby correspond to a different observable basis in $\ca^{\rm sa}$. For $g\in\cg^{\rm or}$ we denote this action on $\hat{O}\in\ca^{\rm sa}$ as $T_g(\hat{O})$. It is also clear that this defines a group homomorphism
\ba
T_{g_1}\circ T_{g_2}=T_{g_1g_2}\,.\label{homo}
\ea

More importantly, $\cg^{\rm or}$ will act on the mathematical descriptions $\varphi$ of the observables, each of which associates to each element in $\ca^{\rm sa}$ a matrix representation. Indeed, denote by $\varphi_e$ the description of $\ca^{\rm sa}$ relative to frame $\{e_\ca\}$. 
Then the local covariance of physics above can equivalently be read as the condition that
\ba
\varphi_e(\hat{O})=\varphi_{\phi_g(e)}(T_g(\hat{O})),\q\q\q\q\forall\,\hat{O}\in\ca^{\rm sa}\,.\label{encodact}
\ea
In other words, the mathematical description of an observable corresponding to some measurement device, which is oriented in some specific way with respect to $\{e_\ca\}$, is identical to that of the transformed observable, which corresponds to an identical measurement device, which is oriented in the same way, but relative to the transformed frame $\{\phi_g(e_\ca)\}$.

\subsection{Relation between spacetime frame transformations and operational group}
\label{sec:relloc}

We thus have two groups, the operational group $\cg^{\rm op}$, which arises from the observation in sec.\ \ref{sec_lmp} that all local frames are fundamentally made up of quantum systems, and the group $\cg^{\rm or}$ of internal frame orientation preserving transformations, which originates in the spacetime properties of local frames defined through $H$ in sec.\ \ref{sec_wald}. Both groups, by construction, relate different local frame orientations -- and thereby the respective descriptions of the local physics -- and both act on quantum state spaces and observables. It is hence pertinent to inquire about the relation of $\cg^{\rm op}$ and $\cg^{\rm or}$.

We begin by arguing that $\cg^{\rm or}$ must act {\it faithfully} on the parent subalgebra $\ca_{\rm p}=\ca_1\otimes\cdots\otimes\ca_n$ which determines an operator basis of the full quantum laboratory due to the LMP, as defined in sec.\ \ref{sec_lmp}. To this end, notice first that $\cg^{\rm or}$ acts on the full operator algebra $\ca_{\rm tot}$ of the entire laboratory. This is clear from (\ref{encodact}) applied to $\ca_{\rm tot}$ because a transformation of the entire laboratory at once that leaves the operational orientation intact must leave the description of the local quantum physics inside it invariant. In consequence, $\cg^{\rm or}$ must also act on $\ca_{\rm p}$ since a change of operator basis for $\ca_{\rm p}$ is equivalent to a change of operator basis for $\ca_{\rm tot}$.

Denote by $T_g$ the action of $g\in\cg^{\rm or}$ on $\ca_{\rm p}$ and suppose the group  $\cg^{\rm or}$ did {\it not} have a faithful action on $\ca_{\rm p}$. Using (\ref{homo}), this means
\ba
\exists\,\,g_1,g_2\in\cg^{\rm or},\q \q g_1\neq g_2,\q\text{but}\q T_{g_1}=T_{g_2}.
\ea
Although $g:=g_1g_2^{-1}\neq e$, where $e$ denotes the unit element in $\cg^{\rm or}$, we then have
\ba
T_{g}=T_{g_1}\circ T_{g_2^{-1}}=T_{g_2}\circ T_{g_2^{-1}}=\text{Id}=T_e\,.
\ea

Next, consider the frames $\{e_\ca\}$ and $\{\phi_{g}(e_\ca)\}$, which are {\it not} aligned.\footnote{By construction, $\cg^{\rm or}$ acts faithfully on the set of frames with same operational orientation at $x\in\cm$.}  In particular, the measurement apparatuses in the two frames are not aligned with one another because the two frames cannot differ only in their time directions $e_0,e'_0$ (see sec.\ \ref{sec_wald}). Using the covariance arguments and, specifically, (\ref{encodact}), this leads to
\ba
\varphi_e(\hat{O})=\varphi_{\phi_{g}(e)}\left(T_{g}(\hat{O})\right)=\varphi_{\phi_{g}(e)}(\hat{O})\,,\q\q\q\q\forall\,\hat{O}\in\ca_{\rm p}^{\rm sa}\,.
\ea
Hence, although {\it not} aligned, the two frames see and describe all observables of the ``parent'' subalgebra in exactly the same way. In particular, although their measurement devices are not aligned, their mathematical descriptions are identical $\varphi_e\equiv \varphi_{\phi_{g}(e)}$.

Consequently, given that the parent ensemble implies the description of all other quantum matter subsystems of the laboratory, it would be impossible to operationally distinguish the frames $\{e_\ca\}$ and $\{\phi_{g}(e_\ca)\}$ through their observations of the local quantum matter physics. On the other hand, classically they {\it could} be distinguished, given their non-alignment. However, owing to our assumption of the local universality of quantum theory, also the local classical physics must emerge from the local quantum physics and so we have a contradiction. We conclude that the local isometries must act faithfully on $\ca_{\rm p}$.

We can now also argue that the action of $\cg^{\rm or}$ on $\ca_{\rm p}$ is contained in that of $\cg^{\rm op}$, which is its maximal (operational) symmetry group. In particular, it thereby acts linearly on $\ca_{\rm p}$. Indeed, from the covariance arguments it follows that every change of frame $\{\phi_g(e_\ca)\}$ for $g\in\cg^{\rm or}$ implies a transformation of an operator basis of $\ca_{\rm p}$. For example, a complete set of measurement devices on the ``parent'' subsystem that is oriented relative to $\{e_\ca\}$ will correspond to a basis of $\ca_{\rm p}$ and changing the local frame to $\{\phi_g(e_\ca)\}$ will also transform the complete set of measurement devices and thus induce a corresponding change of basis in $\ca_{\rm p}$. But every change of operator basis in $\ca_{\rm p}$ is a transformation that is contained in $\cg^{\rm op}$ and, in particular, linear (see sec.\ \ref{sec_lmp}). 

Thus, $\cg^{\rm or}$ is a subgroup of $\cg^{\rm op}$. Our goal is to explore the relation of the two groups even further, to obtain additional information about $\cg^{\rm or}$ by using now also the LMP. By doing so, we will invoke the following plausible mathematical conjecture:
	\begin{conjecture}
		\label{ConjTrans}
		Let $\mathcal{H}$ be a closed connected matrix subgroup of ${\rm SL}(n,\mathbb{C})$ that acts transitively on the non-negative $n\times n$ rank-one observables by conjugation, i.e.\ on $S=\{|\psi\rangle\langle\psi|\,\,|\,\, \psi\in\mathbb{C}^n\setminus\{0\}\}$, such that $X\in\mathcal{H}$ maps $P\in S$ to $XPX^\dagger$. Then,
		\begin{itemize}
			\item if $n$ is odd, we must have $\mathcal{H}={\rm SL}(n,\mathbb{C})$;
			\item if $n$ is even, we either have $\mathcal{H}={\rm SL}(n,\mathbb{C})$ or $\mathcal{H}={\rm Sp}(n,\mathbb{C})$.
		\end{itemize}
		Note that ${\rm SL}(2,\mathbb{C})={\rm Sp}(2,\mathbb{C})$, so for $n=2$ both possibilities coincide.
	\end{conjecture}
	While we do not currently have a proof of this conjecture, it is plausible for the following reasons. First, consider the ``square root'' of this conjecture:  if $\mathcal{H}$ is a closed connected matrix subgroup of ${\rm SL}(n,\mathbb{C})$ that acts linearly and transitively on $\mathbb{C}^n\setminus\{0\}$, then we either have $\mathcal{H}={\rm SL}(n,\mathbb{C})$ or (if $n$ is even) $\mathcal{H}={\rm Sp}(n,\mathbb{C})$. This is consistent with the results in \cite{geatti2012some} for groups transitive on $\mathbb{R}^{2n}\setminus\{0\}$. Furthermore, the compact subgroups ${\rm SU}(n,\mathbb{C})$ resp.\ ${\rm Sp}(n)$ are the unique ones that are transitive on the subset of \emph{normalized} rank-one observables, i.e.\ on the projective space $\{|\psi\rangle\langle\psi|\,\,|\,\,\|\psi\|=1\}$, as shown in \cite{onishchik1993lie}. Our conjecture corresponds to a natural ``unnormalized'' version of that theorem.
	
	We will now use the technical form of the LMP to show the following:
	\begin{thm}\label{thm_gopgor}
		Let $\cg^{\rm or}_0$ be the connected component at the identity of $\cg^{\rm or}$. If $\mathbb{R}_+$ is a subgroup of $\cg^{\rm or}_0$, we define $\cg^{\rm or}_1$ via $\cg^{\rm or}_0=\cg^{\rm or}_1\times\mathbb{R}_+$, and otherwise $\cg^{\rm or}_1:=\cg^{\rm or}_0$ (for an operational interpretation see comment below).
		
		Under the physical background assumptions that lead to Observation~\ref{ObsGop}, we find that either $\cg^{\rm or}_1=\{\mathbf{1}\}$ (the trivial group), or $\cg^{\rm or}_1\supseteq {\rm SO}(3)$. Under the modified physical background assumptions (i), (ii) and (iii) that lead to Observation~\ref{ObsGop2}, and upon invoking Conjecture~\ref{ConjTrans}, we find that either $\cg^{\rm or}_1=\{\mathbf{1}\}$ or $\cg^{\rm or}_1\supseteq {\rm SO}^+(3,1)$, the proper orthochronous Lorentz group.
	\end{thm}
		\begin{proof}
		The proof is given in Appendix \ref{app_gopgor}.	
		\end{proof}

	Let us comment on the possible $\mathbb{R}_+$ subgroup of $\cg^{\rm or}_0$. As $\cg^{\rm or}_0$ is a subgroup of $\cg^{\rm op}$, this would correspond to maps of the form $A\mapsto \lambda A$, where $\lambda>0$. The case $\lambda\neq 1$ is only possible under the alternative set of physical background assumptions (i), (ii) and (iii). Any observer (regardless of their choice of operator basis) can agree on whether a given transformation is simply a scaling of observables; similarly, every observer can agree on whether a map of the form $A\mapsto Y A Y^\dagger=:A'$ has $\det Y=1$. Namely, if there is any invertible linear map $X$ such that $B=X A X^\dagger$ maps descriptions $A$ of observables to descriptions $B$, then in $B$-description the map acts as $B'=(XYX^{-1})B(XYX^{-1})^\dagger$, and a similar result applies if an additional sign change and/or transposition relates the descriptions, as in Observation~\ref{ObsGop2}. That is, the fact that conjugation matrices have unit determinant (up to a global phase factor) is a basis-independent statement. Hence, $\cg^{\rm or}_1$ can be defined as those maps from $\cg^{\rm or}_0$ that can be written in the form $Y\bullet Y^\dagger$ with $\det Y=1$, and the definition of this subgroup is basis-independent. It is clear that the connected group $\cg^{\rm or}_0$ must factorize into a product of $\cg^{\rm or}_1$ and $\mathbb{R}_+$ if it contains non-trivial scalar multiples whatsoever.

The idea of the proof of Theorem~\ref{thm_gopgor} is to consider the action $X\bullet X^\dagger$ of $\cg_0^{\rm or}$-elements on the parent subalgebra. If the orbits of this action allow observers to distinguish operator bases that would otherwise be operationally indistinguishable (i.e.\ equivalent), we have a violation of the LMP. It follows that $\cg_0^{\rm or}$ must either be trivial or ``large enough'' to prevent this from happening. But conjugations with multiples of the identity, $X=\lambda\mathbf{1}$, only lead to a linear scaling that does not tell observers anything new: they know already from the linearity structure what it means to scale observables. This is why the LMP does not tell us anything about a possible subgroup $\mathbb{R}_+$, and this subgroup has to be divided out in the formulation of Theorem~\ref{thm_gopgor}. The group $\mathbb{R}_+$ may or may not be a subgroup of $\cg^{\rm or}$ -- both is compatible with the LMP, and this will be crucial for the case of inhomogeneous dispersion relations in sec.~\ref{sec_finsler}.

Similarly, the technical details of the proof of Theorem~\ref{thm_gopgor} do not allow us to say anything about \emph{disconnected} components of $\cg^{\rm or}$ (for example, space inversion and time reversal if $\cg^{\rm or}$ is the Lorentz group). But this is acceptable, since the operational nature of the corresponding symmetry (for example whether it is broken) may easily depend on the details of the underlying physics, as the well-known example of parity violation in our universe illustrates.

Theorem~\ref{thm_gopgor} allows for the possibility that $\cg_1^{\rm
or}=\{\mathbf{1}\}$ -- how can we understand this trivial case? Consider
the following possibility, described in the terminology of
Subsection~\ref{sec_wald}: suppose that spacetime geometry allows for a
multitude of different frames $\{e_A\}$ resp.\ $\{e'_A\}$, but that
\emph{none} of these frames have the same operational orientation. That
is, \emph{all} spacetime frames can be distinguished via speakable
information; every choice of reference frame can be communicated from
one agent to any other via classical communication (e.g.\ over a
telephone) even if the two agents have never exchanged any physical
systems (like gyroscopes etc.) before. In this case, the group of
internal frame orientation preserving transformations $\cg^{\rm or}$,
and thus $\cg_1^{\rm or}$, will clearly be trivial, i.e.\ equal to
$\{\mathbf{1}\}$. Now, even in such a world, it is conceivable that
there are quantum observables with some spacetime interpretation. For
example, think of a construction plan for a measurement device (e.g.\
something comparable to a Stern-Gerlach device), such that its
construction in spacetime frame $\{e_A\}$ makes it measure observable
$\hat O$, while in frame $\{e'_A\}$ it corresponds to a measurement of
$\hat O'$. But suppose that $\hat O$ and $\hat O'$ are \emph{never}
equivalent in the sense of the technical form of the LMP. For example,
under our first set of physical background assumptions of
Subsection~\ref{sec_lmp}, it means that $\hat O$ and $\hat O'$ are never
unitarily equivalent.

Such a case is fully compatible with the technical form of the LMP:
\emph{spacetime does \textbf{not} allow the observers to distinguish any
observables that would otherwise be indistinguishable}. While spacetime
allows observers indirectly, for example, to operationally distinguish
observables $\hat O$ and $\hat O'$ (as those that correspond to
speakably distinguishable frames $\{e_A\}$ and $\{e'_A\}$,
respectively), this is not relevant for the LMP because $\hat O'\neq
U\hat O U^\dagger$ for all unitaries $U$. In other words, $\hat O$ and
$\hat O'$ \emph{are already operationally distinguishable anyway} (due
to our physical background assumptions) by means of their eigenvalues.
They are not ``equivalent'' in the sense of the technical form of the
LMP, and thus the LMP is not violated.

This trivial case (spacetime does not carry any symmetry, and does not
break any quantum symmetry either) will not be considered further in the
following.

In summary, the LMP essentially implies that the internal frame
orientation preserving transformations $\cg^{\rm or}$ (as long as its
connected component at the identity is not trivial -- a case we shall
ignore from now on) must either contain at least the rotations, or even
the Lorentz group, depending on our choice of physical background
assumptions. Furthermore, $\cg^{\rm or}$ must be a subgroup of $\cg^{\rm
op}$, but the two groups need not be identical.

\subsection{What does it mean if the Local Mach Principle is violated?}\label{sec_nolmp}

Let us first recapitulate what we have done so far to then interpret what a violation of the Local Mach Principle could mean. The colloquial form of the LMP in sec.~\ref{sec_0lmp}, originates in a non-trivial assumption: the effective spacetime structure in which the local inertial laboratory resides is the coarse-grained, large-scale limit of a {\it special} (universality) class of quantum gravity states such that the net interaction of the matter inside the laboratory with quantum gravitational degrees of freedom is zero on average. In other words, 
 any {\it direct} interaction of the matter with quantum gravitational degrees of freedom has been washed out through renormalization at the relevant laboratory scales, which we consider in our thought experiment. In consequence, the local inertial frame, at the relevant scales, is not only isolated from the matter outside the frame, but also from the effective quantum gravitational degrees of freedom. This already implies that the local effective spacetime structure in the local inertial frame must appear completely isotropic and not offer any structure for the observer to orient themselves. Instead, it is only the quantum matter physics in the local laboratory that is left over to `self-generate' any reference structure for the observer to orient their frame; this is the physical content of the LMP. Owing to the assumed universality of quantum theory, the local inertial frame is thereby truly a quantum reference frame and self-sufficient.

Like the Global Mach Principle, there are
different possible ways of formulating this intuitive principle as a
concrete mathematical postulate, and different formulations will entail
different consequences.
In this paper, we have suggested a specific technical form of the LMP that is based on further assumptions. Most drastically, we assumed in sec.\ \ref{sec_lmp} that {\it discrete}, i.e.\ finite-dimensional degrees of freedom fully encode information about spacetime orientation and that the agent can restrict to those. In particular, this assumes that the observer can separate these discrete degrees of freedom, e.g., from the momentum modes. While this greatly simplified our technical discussion, we conjectured that this finite-dimensionality assumption can actually be dropped without modification of the main conclusions. We noted that this assumption is also satisfied in Minkowski space.

The concrete consequences of the LMP then depended on the physical background assumptions to which we exposed the agent. The physical background assumptions leading to Observation~\ref{ObsGop} assume the
\emph{eigenvalues} of observables to be frame-independent,
and can thus be used by observers as extra data to construct a frame. In
this case, different frames will differ by a unitary transformation, and Theorem~\ref{thm_gopgor} proves that
the internal frame orientation preserving group $\mathcal{G}^{\rm or}$ contains \emph{at least}
${\rm SO}(3)$. On the other hand, in Observation~\ref{ObsGop2}, we give
the consequences of the LMP if we allow eigenvalues to depend on the
frame, but only assume that the zero eigenvalues have frame-independent
significance (motivated by measurement devices like Stern-Gerlach
devices and outcomes that carry physical units). In this case, Theorem \ref{thm_gopgor} shows that  $\cg^{\rm or}$ must be even larger if the LMP holds true,
containing \emph{at least} the orthochronous Lorentz group ${\rm SO}^+(3,1)$. Notice that both these results are consistent with the conceptual observation above that the LMP implies space or spacetime to be locally completely isotropic, given the large local frame symmetry groups for a $3+1$ spacetime. However, this does not yet imply that space and spacetime will locally be Euclidean and Lorentizan metric geometries, respectively. We will investigate in detail which constraints these results impose on spacetimes with dispersion relations in sec.\ \ref{sec_finsler} shortly.

This discussion also directly suggests how to interpret spatiotemporal structures in which we find the LMP to be violated, as in some spacetimes with dispersion relations in sec.\ \ref{sec_finsler} below. Of course, any of our more detailed technical assumptions might, in principle, fail for matter quantum theory in arbitrary effective spacetimes. More generally, however, we may interpret a violation of the LMP as the effective spacetime environment corresponding to the large-scale limit of quantum gravity states in which the net interaction of the matter with the quantum gravitational degrees of freedom is {\it not} washed out entirely through renormalization. That is to say, violations of the LMP may indicate a non-vanishing net coupling of the matter degrees of freedom to effective quantum gravitational ones even at the relevant laboratory scales. As such, the matter subsystems in the local laboratory are not maximally isolated and may become indirectly correlated via each of their direct interactions with the effective quantum gravitational degrees of freedom. This situation can be qualitatively modelled by Example~\ref{ExHeisenberg} with $h\neq 0$ and interpreting $h$ in a role analogous to how we envisage the effective spacetime structure here. The spins can become correlated through their direct interaction with $h$ even if one switched off $J$. While this magnetic field
	is ultimately quantum too, it is only its effective classical
	description that enters the effective Hamiltonian as a parameter $h$ and
	a distinguished direction $z$. 
	
The agent could then exploit that these effective quantum gravitational degrees of freedom generate a local effective spacetime anisotropy even in local inertial frames to facilitate their task to define a reference frame orientation. This local non-trivial effective spacetime structure would not be there if the LMP was satisfied. It is as if the matter in the local laboratory sees a proper vacuum of the matter outside the lab, but not an effective vacuum of the quantum gravitational degrees of freedom. As such, the local laboratory would be inertial in a pure matter sense, but not in a quantum gravitational sense; the net `force' of quantum gravitational degrees of freedom onto the matter is non-zero.

In this light, we may interpret deviations from the symmetries $\cg^{\rm or}_1\supseteq {\rm SO}(3)$ or $\cg^{\rm or}_1\supseteq {\rm SO}^+(3,1)$, arising as a consequence of the LMP, as effective quantum gravity effects. This is, at least conceptually, consistent with efforts in quantum gravity phenomenology that deal with Lorentz violations \cite{Jacobson:2005bg,	AmelinoCamelia:2008qg}.	We also note that there are other efforts in quantum gravity phenomenology which investigate possible interactions between low-energy matter physics and effective quantum gravity degrees of freedom. For instance, a recent approach seeks to explain the cosmological constant as emerging from the diffusion of energy of matter physics into Planck scale granularity \cite{Josset:2016vrq, Perez:2017krv, Perez:2018wlo}.

\section{The LMP and dispersion relations}\label{sec_finsler}

Dispersion relations are viability conditions, which the four-momenta of physical point particles or field modes have to satisfy. They encode the causal structure of spacetime, define observer directions as well as the observers' spatio-temporal splits of spacetime and the geometry of spacetime, i.e.\ the gravitational interaction \cite{Raetzel:2010je,Barcaroli:2015xda}. In modern physics, dispersion relations emerge mainly as the point particle limit of field theories (technically the principal symbol of partial differential field equations) \cite{Raetzel:2010je}, for example, in the study of premetric or area-metric electrodynamics \cite{Rubilar:2007qm,Punzi:2007di}, as a tool in the study of possible violations of fundamental local Lorentz invariance \cite{Liberati:2013xla,Mattingly:2005re}, or in effective approaches to quantum gravity \cite{AmelinoCamelia:2008qg}.

In this section, we seek to investigate the implications of the LMP on spacetimes defined through dispersion relations. To this end, we technically formulate dispersion relations as level sets of Hamilton functions. Subsequently, we specify observer frames in terms of how they `see' dispersion relations; in particular, we compare mass-shell encodings of different observers, which will constitute the structure necessary for the {\it operational orientation} generally introduced in sec.\ \ref{sec_wald}.

\subsection{Dispersion relations as Hamilton functions}\label{sec:DispHam}

Technically, dispersion relations are implemented on a manifold $\mathcal{M}$ in terms of Hamilton functions $H$ on the manifold's cotangent bundle, i.e.\ physically speaking on its point particle phase space. In general, all of the following  can be studied in any dimension, however, we restrict to four spacetime dimensions here for simplicity. In local coordinates, an element $K$ on the cotangent bundle, i.e.\ a $1$-form on spacetime in some cotangent space $T^*_x\mathcal{M}$, can be expressed as $K = k_\mu \dd x^\mu = (x,k)$. These are called manifold induced coordinates of the cotangent bundle. The Hamilton function is a map
\begin{align}
H: T^*\mathcal{M} \rightarrow \mathbb{R},\quad K \mapsto H(K)\,,
\end{align}
which in local manifold induced coordinates reads $H(K) = H(x,k)$. In what follows we employ the coordinate representation. The level sets of the Hamilton functions $H(x,p) = const \geq 0$ represent the dispersion relations, which the particles have to satisfy and its Hamiltonian equations of motion $\dot x^\mu = \bar{\partial}^\mu H,\ \dot k_\mu = - \partial_\mu H$ determine the particles' trajectories. We use the abbreviations $\partial_\mu = \frac{\partial}{\partial x^\mu}$ and $\bar{\partial}^\mu = \frac{\partial}{\partial k_\mu}$ for the appearing partial derivatives, respectively. 

In order to identify Hamiltonians which define a relativistic spacetime structure, the following minimalistic criterion is employed here. Consider the Hessian of the Hamiltonian with respect to the covector coordinates, also called the Hamilton metric
\begin{align}
	g^{H\mu\nu} (x,k)= \frac{1}{2}\bar{\partial}^\mu\bar{\partial}^\nu H(x,k)\,.
\end{align}
We denote manifolds equipped with a dispersion by the tuple $(\mathcal{M},H)$ and call them \emph{Hamiltonian spacetimes} if  $\forall x \in \mathcal{M}$ there exists a connected component $C_x \subset T^*_x\mathcal{M}$ such that  on $C_x$ the signature of $g^H$  is Lorentzian $(+,-,-,-)$ and $H(x,k) > 0$. Moreover $C := \bigcup_{x\in M}C_x$ shall be a smooth sub-bundle of $T^*\mathcal{M}$. The connected component $C_x$  is interpreted as the set of physically viable massive momenta at $x$. This definition of Hamiltonian spacetimes includes spacetimes with bi-hyperbolic polynomial dispersion relations \cite{Raetzel:2010je}, where the set of massive momenta is a convex hyperbolicity cone of the dispersion relation, and spacetimes with inhomogeneous dispersion relations employed in effective models of quantum gravity \cite{Barcaroli:2015xda}. Finally, the in general momentum dependent geometry of spacetime can be derived from $H$ and its derivatives, in a similar way to deriving the geometry of spacetime from a metric on a pseudo-Riemannian spacetime \cite{Barcaroli:2015xda,Miron}.
 
The approach to the geometry of spacetime in terms of a Hamilton function $H$ on the cotangent bundle is dual to a Finslerian spacetime geometry, derived from a Finsler Lagrangian $L$ on the tangent bundle of spacetime \cite{Javaloyes:2018lex,Pfeifer:2011tk,Beem}, in case a Legendre map exists which maps the Hamiltonian to a Finsler Lagrangian. Such a map is constructed from the Helmholtz action for point particles, which ensures that freely falling particles satisfy the dispersion relation induced by the $H$
\begin{align}
S_H[x, k, \lambda] = \int d\tau\ (\dot x^\mu k_\mu + \lambda f(H(x,k)))\,.
\end{align}
Here, $f(H)$ is chosen such that $f(H) =0$ implements the desired dispersion relation $H = const \geq 0$. Solving the equations of motion for $\lambda$ and $k$, respectively, allows one to obtain an equivalent point particle action defined by a Finsler Lagrangian $L$
\begin{align}
S[x] := S_H[x, k(x,\dot x), \lambda(x,\dot x)] = \int d\tau\ L(x,\dot x)\,.
\end{align}
This construction has been employed for bi-hyperbolic polynomial dispersion relations \cite{Raetzel:2010je}, for example obtained from premetric electrodynamics \cite{Gurlebeck:2018nme}, as well as for the $\kappa$-Poincar\'e dispersion relation \cite{Letizia:2016lew}.

To investigate the relation between the LMP and the geometry of spacetime, we are mainly interested in the local properties of the dispersion relation defining Hamiltonian. Therefore, our attention lies on the function $H_x(k) = H(x,k)$. In addition, we need to extend the concept of symmetries of a Hamiltonian \cite{Barcaroli:2015xda} to local symmetries, i.e.\ to diffeomorphisms $\Psi:T^*_x\mathcal{M}\rightarrow T^*_x\mathcal{M}$ such that $H_x(\Psi(k)) = H_x(k)$. Infinitesimally, they can be described by vector fields $\xi = \xi_\mu(x,p) \bar{\partial}^\mu$, along which the Hamiltonian is constant $\xi(H) = 0$. Observe that every Hamiltonian possesses local symmetries induced by the vector fields $\xi^{\mu\nu} = \bar{\partial}^\mu H \bar{\partial}^\nu - \bar{\partial}^\nu H \bar{\partial}^\mu$, since trivially $\xi^{\mu\nu}(H)=0$.

More interestingly, we now come to two core questions of this section: we wish to characterize dispersion relations, which feature an invariance under either local Lorentz transformations or purely spatial rotations. This will help us later, in Sec.~\ref{sec:implLMP}, to translate the implications of the LMP, as expressed in Theorem \ref{thm_gopgor}, into non-trivial constraints on the spacetime structure defined through dispersion relations. We begin with local Lorentz symmetry.

\begin{thm}[Local Lorentz invariant dispersion relations]\label{thm:LLI}
	Consider a Hamiltonian spacetime $(\mathcal{M},H)$ and let $g$ be some Lorentzian spacetime metric. The generators of local Lorentz transformations on the cotangent spaces of spacetime, $M^{\mu\nu} = g^{\mu\sigma}k_\sigma \bar{\partial}^\nu - g^{\nu\sigma}k_\sigma \bar{\partial}^\mu$, generate local and linear symmetries of $H$ if and only if $H_x(k) = h_x(w(k))$, where $w(k) = g^{\mu\nu}(x)k_\mu k_\nu$ and $h_x(w)$ is a function in one variable only.
\end{thm}
A similar answer can be given to the question asking for dispersion relations possessessing a rotational symmetry. To do so, we decompose the cotangent bundle in the following way. Let $\{k^\mu\}_{\mu=0}^3$ be manifold induced coordinates on $T^*_x\mathcal{M}$ and let $\Sigma_x \subset T_x\mathcal{M}$ be a three-dimensional sub-vector space of $T^*_x\mathcal{M}$. Then there exist linear combinations of the original coordinates $p_A = A^\mu{}_A(x) k_\mu$ such that $\{p_A\}_{A=1}^3$ are coordinates of $\Sigma_x$. A positive or negative definite scalar product $s$ in $\Sigma_x$ then allows us to classify dispersion relations which are invariant under rotations in $\Sigma_x$.

\begin{thm}[Rotationally invariant dispersion relations]\label{thm:LRI}
	Consider a Hamiltonian spacetime $(\mathcal{M},H)$ and let $\Sigma_x\subset T^*_x\mathcal{M}$ be a three-dimensional sub-vector space of $T^*_x\mathcal{M}$ equipped with a positive/negative definite scalar product $s$ and coordinates $\{p_A = A^\mu{}_A(x) k_\mu\}_{A=1}^3$.	The generators of orthogonal transformations in $(\Sigma_x,s)$ are $M^{AB} = s^{AC}p_C \bar{\partial}^B - s^{BC}p_C \bar{\partial}^A$. They generate local and linear symmetries of $H$ if and only if $H_x(k) = r_x(p_0,v(p))$, where $p_0$ completes the sub-vector space coordinates $\{p_A\}_{A=1}^3$ to coordinates on $T^*_xM$, $v(p) = s^{AB}(x)p_Ap_B$ and $r_x(p_0,v)$ is a function in two variables only.
\end{thm}
\noindent The proofs of the theorems can be found in Appendix \ref{app:thmLLI}. 

Note that we do {\it not} assume in our discussion that the Hamiltonian spacetime $(\mathcal{M},H)$ features a metric structure. But the two theorems express the fact that \emph{if} the Hamiltonian has either a linear and local Lorentz symmetry or rotational invariance, defined by the symmetry vector fields $M^{\mu\nu}$ resp.\ $M^{AB}$, then the spacetime \emph{must} feature a four-dimensional Lorentzian or three-dimensional Euclidean metric as building block of the Hamiltonian, respectively. This, however, does not yet imply that, e.g., in the first case, the spacetime is a standard Lorentzian metric geometry, i.e.\ that the Hamiltonian is simply given by $H=g^{-1}(k,k)$; more general dispersion relations with local Lorentz invariance exist, in particular, if $H$ is inhomogeneous. 

An example of a locally Lorentz invariant dispersion relation is clearly the general relativistic dispersion relations, which we discuss in Sec.~\ref{ssec:LLI}; a non-trivial example for a locally rotationally invariant dispersion relation is the $\kappa$-Poincar\'e dispersion displayed in Sec.~\ref{ssec:kappa}.

Next we discuss the relation between observers who obtain an identical mass-shell and the existence of local linear symmetries.

\subsection{Observers and their mass shell encodings}\label{sec:ObsDispRleations}

In contrast to Sec.~\ref{sec_oper}, we will now describe observer frames equivalently in terms of co-tetrads instead of tetrads, as this is more convenient for dealing with dispersion relations. A co-tetrad $\{\theta^a = \theta^a{}_\mu dx^\mu\}_{a=0}^3$ is a basis of the cotangent spaces $T_x^*\mathcal{M}$ of $\mathcal{M}$. An observer co-tetrad $\{\hat \theta^a\}_{a=0}^{3}$ is a co-tetrad which satisfies the following conditions determined by the Hamiltonian
\begin{align}\label{eq:obsdef}
\hat \theta^0 \in C_x,\quad H_x(\hat \theta^0) = M > 0,\quad \hat \theta^A{}_\mu \bar{\partial}^\mu H(x,\hat \theta^0) = 0\,,\ A=1,2,3\,,
\end{align}
where $M$ is a constant. These conditions are general properties of observers, which are extended by more} precise observer definitions in the case of hyperbolic polynomial dispersion relations \cite{Raetzel:2010je} and quantum gravity phenomenology approaches~\cite{Barcaroli:2017gvg}.\footnote{For bi-hyperbolic polynomials, the relevant connected component $C_x$ is a hyperbolicity cone of the dispersion relation \cite{Raetzel:2010je,Pfeifer:2011tk}. To illustrate the need for a further specified definition of observers, we can consider a specific example; a dispersion relation given by a homogeneous fourth order hyperbolic polynomial. In the generic case of such a polynomial, there exist two forward and two backward light cones; the polynomial's vanishing set splits into four conical surfaces that intersect at the origin. In between these surfaces, the dispersion relation has a fixed sign. Therefore, three connected cones of co-vectors exist that sattisfy the first condition in (11), the positivity condition. Only those co-vectors can be regarded as valid observers that are lying in one of the hyperbolicity cones, the cones of covectors that are co-normals to valid initial data hypersurfaces of the linear matter field equations giving rise to the dispersion relation. For the $\kappa$-Poincar\'e dispersion relation, the observer momenta do not form a cone and it suffices to define them via the sign of the dispersion relation. The set $C_x$ is just given by the momenta, satisfying $H_x(\hat \theta^0)>0$ and $\theta^0{}_0>0$. Hence, the precise identification of the set of massive momenta $C_x$ depends on the class of dispersion relations which are being investigated.}

The first condition implies that $\hat \theta^0$ is the momentum of a massive particle trajectory, the second condition states that the tangent of the observer trajectory $\dot x^\mu = \bar{\partial}^\mu H(x,\hat \theta^0)$ is co-normal to the spatial co-tetrads. The dual tetrad $\{\hat e_a\}_{a=0}^3$ of the co-tetrad satisfies $\hat \theta^a(\hat e_b) = \delta^a{}_b$ and describes an observer on the tangent spaces of spacetime, as in Sec.~\ref{sec_wald}, where $\hat e_0 = \dot x$ is the tangent of the observer's trajectory.

Having defined the observer co-tetrads, we can consider the dispersion relation expanded into these physically meaningful bases, instead of in local coordinates. To this end, we expand $T_x\mathcal{M}\ni K = k_\mu \dd x^\mu = \hat k_a \hat \theta^a$, which yields $k_\mu(\hat k) = \hat \theta^a{}_\mu \hat k_a$ and write
\begin{align}
H_x(k(\hat k)) = H_x(\hat \theta^a{} \hat k_a) := H_{\hat \theta, x}(\hat k)\,.
\end{align}
In other words, as outlined in Sec.~\ref{sec_wald}, the co-tetrad of an observer defines a canonical observer encoding map $f_{\hat \theta}$ which assigns to a physical quantity the values an observer will measure.  For the momentum $1$-form $K$ this map is $f_{\hat \theta}(K)_a = \hat k_a$. We point out that the observer encoding map is defined without involving any additional structure, apart from the observer co-tetrad themselves. In particular, this map assumes that the components of $K$ with respect to the observer co-tetrad are directly the momentum components which the observer would measure. This interpretation of the $\hat{k}_a$ need not be justified in general. If not, this map would have to be constructed differently by using additional geometric structures. In the following, we will, however, always assume this simple form of the observer encoding map $f_{\hat \theta}$.

Notice that the Hamilton function, expressed in an observer basis $H_{\hat \theta, x}$, is a map from $\mathbb{R}^4$ to $\mathbb{R}$. We shall refer to it as the encoding of the dispersion relation at $x$ with respect to the observer co-tetrad $\hat \theta$ and to $H_{\hat \theta, x}(\hat k) =  const = m^2$ as the observer's encoding of the mass shell.  In the following, the mass shell encoding will take the role of the sufficient (``speakable'') structure, provided by the spacetime structure $H$, to meaningfully state whether two co-tetrads have the same operational orientation. This concretizes the general discussion and arguments of Sec.~\ref{sec_wald} to spacetimes with dispersion relation $H$. 

Next, consider two different observer co-tetrads $\{\hat \theta^a\}_{a=0}^3$ and $\{\tilde \theta^a\}_{a=0}^3$. By definition they are both bases of $T^*_x\mathcal{M}$ and thus can be expressed in terms of each other via
\begin{align}
	\hat \theta^a =  \Lambda^a{}_b \tilde \theta^b\quad \mathrm{ and }\quad \hat e_a = (\Lambda^{-1})^b{}_a\tilde e_b\,.
\end{align}
The components of the matrices can be projected as
\begin{align}\label{eq:obstrans}
\Lambda^a{}_c = \tilde e_c(\hat \theta^a)\quad \mathrm{ and }\quad  (\Lambda^{-1})^c{}_a = \tilde \theta^c(\hat e_a)\,.
\end{align} 
Observe that in general $\Lambda$ is a non-linear function of the co-tetrads, i.e.\ $\Lambda = \Lambda(\hat \theta, \tilde{\theta})$ in a non-trivial way. For the sake of readability we do not display this dependence further. Let us remark that for generic dispersion relations, applying the observer transformation (\ref{eq:obstrans}) to a third observer co-tetrad $\{\check\theta^a\}$ does \emph{not} produce another observer co-tetrad. In other words, the observer transformations $\Lambda^a{}_b$ for all $\{\hat \theta^a\}$ and $\{\tilde \theta^a\}$ do not form a group representation on $\mathbb{R}^4$. This can be seen by considering infinitesimal observer transformations of a given co-tetrad $\{\hat \theta^a\}$. Indeed, in Appendix \ref{sec:obsdep}, we show that the conditions, which an infinitesimal observer transformation has to fulfill, depend on the momentum $\hat \theta^0$ of the start co-tetrad if the dispersion relation has a non-vanishing third derivative with respect to the momenta. The latter is only the case for dispersion relations that are quadratic in the momenta (as, e.g., general relativistic ones) and only in this case do the observer transformation directly define a group representation on $\mathbb{R}^4$.

It is clear that two observers use an identical encoding of the dispersion relation if $H_{\hat \theta, x}$ and $H_{\tilde \theta, x}$ are identical functions on $\mathbb{R}^4$, i.e.\ if the dispersion relation has the same form in the two observer co-tetrads. The following proposition states a condition on the existence of observers with equal dispersion relation encodings. It is related to the existence of local and linear symmetries, which we discussed at the end of Sec.~\ref{sec:DispHam}:

\begin{prop}[Identical mass-shell encodings]\label{prop:ObsMS}
	Let $\{\hat \theta^a\}_{a=0}^3$ and $\{\tilde \theta^a\}_{a=0}^3$ be observer co-tetrads on the Hamiltonian spacetime $(\cm,H)$, which are related by the observer transformation matrix $\Lambda^a{}_b(\hat \theta,\tilde{\theta})$. The observers agree on their encodings of the dispersion relation, $H_{\hat \theta, x}$ and $H_{\tilde \theta, x}$, via the encoding map $f_{\hat \theta}$ and $f_{\tilde \theta}$ if and only if the observer transformation is a local \emph{linear} symmetry of the Hamiltonian.
\end{prop}
\noindent We proof of this proposition in Appendix~\ref{app:propObsMS}

An illustrative way to understand the identical mass shell encodings is the following: Let the two observers with co-tetrads  $\{\hat \theta^a\}$ and $\{\tilde \theta^a\}$ sample their mass shell by measuring energies and momenta of particles with a fixed mass. Then one of them, say the first observer, sends a fit of their mass shell sampling (or even the raw data) by classical communication. The second observer then receives or recovers a representation of the encoding $H_{\hat \theta,x}$ as a function on $\mathbb{R}^4$. If and only if the observer co-tetrads are connected by an observer transformation that induces a local \emph{linear} symmetry of $H_x$ in the sense of Proposition \ref{prop:ObsMS}, will the second observer find that their encoding $H_{\tilde \theta,x}$ coincides with~$H_{\hat \theta,x}$. 

The comparison of dispersion relations, respectively mass-shell encodings, allows us to group observers into equivalence classes. In fact, this is precisely the specification of the operational equivalence classes, defined by the same `operational orientation' in Sec.~\ref{sec_wald}, to Hamiltonian spacetimes. Here, we say that two observer co-tetrads have the same operational orientation, and thus belong to the same equivalence class, if they yield the same encoding of the dispersion relation. This is all the operationally accessible structure that the spacetime structure $H$ offers the observers. Hence, we say that two observer co-tetrads $\{\hat \theta^a\}$ and $\{\tilde \theta^a\}$ belong to the same equivalence class with respect to the encoding of the mass shell and write $\{\hat \theta^a\}\sim\{\tilde \theta^a\}$ if the observer transformation $\Lambda$ between the co-tetrads is a local and \emph{linear} symmetry of the Hamiltonian $H_x$. 
Each equivalence class forms a sub-manifold of the manifold of all observer co-tetrads. The co-dimension of the manifolds of equivalence classes quantifies the amount of information that can be gained by comparing the two mass shells. In other words, if $\mathcal{F}$ denotes the manifold of all observer co-tetrads, the quotient $\mathcal{F}/\sim$ represents the information which the observers can access, i.e.\ the ``speakable information'', which observers can agree on by classical communication. The dimension of the sub-manifolds of equivalence classes quantifies the ``unspeakable information'' about tetrad orientations that observers can{\it not} communicate classically. For these sub-manifolds, which encode the remaining possible observer relations \emph{after} classical communication, we can give the following proposition:

\begin{prop}[Observer transformations and local and linear symmetries]\label{prop:obstrans}
	Let $\cg^{\rm dis}$ be the group of local and linear symmetries of $H_x$. For each co-tetrad $\{\theta^a\}_{a=0}^3$, the map $I_\theta : \cg^{\rm dis} \rightarrow \mathfrak{G}_\theta \subset GL(4)$ with $I_\theta (\Psi)^a{}_b = \Lambda^a{}_b(\Psi(\theta),\theta)$ for all $\Psi\in \cg^{\rm dis}$ is a group isomorphism. Furthermore, $\mathfrak{G}_\theta$ is equivalent to the set of all observer transformations between observer co-tetrads in the same equivalence class as the co-tetrad~$\{\theta^a\}_{a=0}^3$.
\end{prop}

The proof is given in Appendix \ref{sec:proofprop2}.

It follows from Proposition 2 that the remaining operational ignorance about the relation between observer co-tetrads after classical communication is encoded in a group structure. This is precisely the internal frame orientation preserving transformation group $\mathcal{G}^\mathrm{or}$ of Sec.~\ref{sec_wald}, specified to a Hamiltonian spacetime. Furthermore, Proposition 2 says that $\mathcal{G}^\mathrm{or}$ is isomorphic to the group of local \emph{linear} symmetries $\cg^{\rm dis}$ of the Hamilton function $H$. The latter thereby also quantifies the amount of ``unspeakable information''. Notice that these results are subject to our choice of the observer encoding functions $f_{\hat \theta}(K)$ being linear.

\subsection{The three groups $\cg^{\rm op}$, $\cg^{\rm or}$ and $\cg^{\rm dis}$ and their relation}\label{sec_3g}

In this work, we thus have three groups appearing: 
\begin{description}
\item[$\cg^{\rm op}$] This operational group comes directly out of the structure of quantum matter. It derives from the technical implementation of the LMP in Sec.~\ref{sec_lmp} and the premise that local inertial reference frames are fundamentally made up of quantum matter. Recall, in particular, that, owing to our finite-dimensionality assumption in Sec.~\ref{sec_lmp}, this group acts on observables and state spaces of {\it discrete} quantum matter degrees of freedom. 
\item[$\cg^{\rm or}$] This is the internal frame orientation preserving transformation group, generally introduced in Sec.~\ref{sec_wald} and in \ref{sec:ObsDispRleations} specified as the symmetry group of operational equivalence classes of observers who see the same mass shell encoding. Hence, it arises from dispersion relations and acts on {\it continuous} momentum degrees of freedom. At the same time, it follows from Sec.~\ref{sec:relloc} that it also acts on the {\it discrete} quantum matter degrees of freedom. As such, this group is the  connection between {\it continuous} and {\it discrete} degrees of freedom and, specifically, between our discussion of the quantum matter and the effectively classical dispersion relations.
\item[$\cg^{\rm dis}$] This is the group of local and \emph{linear} symmetries of the (encoding of the) dispersion relation $H_x$. In particular, a priori it acts purely on {\it continuous} momentum degrees of freedom.
\end{description}

A priori, their relation is not obvious. However, Theorem \ref{thm_gopgor} and Proposition \ref{prop:obstrans} provide their crucial links. 
\begin{description}
\item[Proposition \ref{prop:obstrans}] proves that $\cg^{\rm dis}$ and $\cg^{\rm or}$ are actually isomorphic. 
\item[Theorem \ref{thm_gopgor}] does {\it not} prove that, in turn, $\cg^{\rm or}$ and $\cg^{\rm op}$ are isomorphic too (recall the discussion at the end of Sec.~\ref{sec:relloc}). However, it uses the LMP to constrain $\cg^{\rm or}$ by exploiting that it must be a subgroup of $\cg^{\rm op}$. The precise technical implications depend on the physical background assumptions. 
\begin{description}
\item[Observation~\ref{ObsGop}:] The physical background assumptions leading to it imply, firstly, that $\cg^{\rm op}\supseteq\rm{SO}(3)$ and, in consequence, through Theorem \ref{thm_gopgor}, that also $\cg_0^{\rm or}\supseteq\rm{SO}(3)$, where $\cg_0^{\rm or}$ is the connected component at the identity. (We recall from Theorem \ref{thm_gopgor} that also the trivial case $\cg_0^{\rm or}=\{\mathbf{1}\}$ appears, which too has a consistent physical interpretation, as discussed at the end of Sec.~\ref{sec:relloc}. But we shall henceforth ignore this rather exotic case.)
\item[Observation~\ref{ObsGop2}:] The physical background assumptions leading to it imply, firstly, that $\cg^{\rm op}\supseteq\mathbb{R}_+\times\rm{O}(3,1)$ and, in consequence, through Theorem \ref{thm_gopgor}, that $\cg_1^{\rm or}\supseteq\rm{SO}^+(3,1)$, where $\cg_1^{\rm or}$ is the connected component at the identity with multiples of the identity, i.e.\ $\mathbb{R}_+$, factored out. (Similarly, Theorem \ref{thm_gopgor} also permits the trivial case $\cg_1^{\rm or}=\{\mathbf{1}\}$, which too has a consistent interpretation, see Sec.~\ref{sec:relloc}, but again we shall henceforth ignore it.)
\end{description}
Recall that the LMP, through Theorem \ref{thm_gopgor}, says nothing about discrete transformations in $\cg^{\rm or}$. However, this is not a caveat for our discussion. As discussed at the end of Sec.~\ref{sec:relloc}, observers can always restrict to the connected component at the identity as this is operationally distinguished. In many spacetime structures, observers could also even agree by classical communication whether their labs are related by discrete transformations such as parity or time reversal. For instance, this is the case in general relativistic spacetimes with Standard Model matter where observers could exploit CP violation to agree on the handedness of their frames or exploit a possible global time orientation to also agree on the time direction. \\
Similarly, the LMP is silent on whether a factor $\mathbb{R}_+$ is actually part of $\cg^{\rm or}$ or not. In the context of dispersion relations, this factor $\mathbb{R}_+$ would amount to rescalings of the mass shell. For homogeneous dispersion relations a factor $\mathbb{R}_+$ would thus be a symmetry of $H_x$, while it would {\it not} be a symmetry for inhomogeneous dispersion relations. As such, \emph{the LMP is a priori compatible with both homogeneous and inhomogeneous dispersion relations}. Notice also that observers in the inhomogeneous case would always be able to agree by classical communication on whether their respective mass shell encodings are related by a rescaling. In the homogeneous case, this depends on some background assumptions. For instance, under the assumption that the same matter theory applies everywhere in the universe, in analogy to the Standard Model in Lorentzian spacetimes, observers might agree by classical communication on what an `electron' would be and agree to use this as a mass standard. In that case, the $\mathbb{R}_+$ relation would, in fact, also be part of the ``speakable'' information. 

\end{description}
	
For operational purposes, it is thus sufficient to restrict to the connected component of $\cg^{\rm or}$ at the identity modulo rescalings, i.e.\ to $\cg^{\rm or}_1$. We will now exploit this to formulate the constraint that the LMP imposes on dispersion relations. To this end, given that $\cg^{\rm or}\simeq\cg^{\rm dis}$, we also write $\cg^{\rm dis}_1$ for the connected component of $\cg^{\rm dis}$ at the identity modulo possible rescalings.
	
\subsection{Implications of the LMP {for spacetimes with dispersion relation}}
\label{sec:implLMP}

Enforcing the LMP on a Hamiltonian spacetime, the above relations, implied by \emph{both} Theorem \ref{thm_gopgor} and Proposition~\ref{prop:obstrans}, thus lead to the following conclusions:
\begin{itemize}
	\item If $\mathcal{G}_0^\mathrm{or}\simeq \cg^{\rm dis}_0\supseteq SO(3)$, complying with Observation \ref{ObsGop}, 
	Theorem \ref{thm:LRI} immediately implies that $H_x$ is a function of two variables, where one of these is defined by a three dimensional Euclidean metric norm of momenta.
	\item If $\mathcal{G}^\mathrm{or}_1\simeq\cg^{\rm dis}_1 \supseteq SO^+(3,1)$, complying with Observation \ref{ObsGop2}, 
	Theorem \ref{thm:LLI} immediately implies that $H_x$ is a function of a four-dimensional Lorentzian metric. If, additionally, $H_x$ is a homogeneous polynomial of even degree (as for example when induced by linear matter field equations with well posed initial data problem \cite{Raetzel:2010je}), Euler's homogeneous function theorem directly implies $H_x = \alpha (g^{-1}_x)^n$, where $g$ is a Lorentzian metric and $\alpha\in\mathbb{R}$ and $n\in\mathbb{N}$ even. Then, $(\mathcal{M},H_x)$ is a Lorentzian spacetime. In other words, \emph{given a well defined initial value problem for a homogeneous dispersion relation, our LMP  entails a Lorentzian metric spacetime.}
\end{itemize}

In the next section, we will discuss three examples of Hamiltonian spacetimes: Lorentzian spacetimes, a spacetime inspired by electrodynamics in a uniaxial crystal and $\kappa$-Poincar\'e dispersion relations. Furthermore, we will interpret the above conclusions for each example in the context of operational ignorance and ``speakable information''.

\subsection{Example dispersion relations}\label{sec:ex}
After the general discussion we explain our findings on specific example dispersion relations and clarify their relation to the LMP. The fundamental assumption of the LMP, that quantum matter decouples completely from the geometry of spacetime is only satisifed in the first example of Lorentzian metric spacetimes. In the uniaxial crystal spacetime two vector fields define the geometry of spacetime as additional structure to the metric. For $\kappa$-Poincar\'e spacetimes one additional vector field to the metric defines the dispersion relation. In the context of quantum, gravity phenomenology, these additional fields effectively describe the interaction of point particles with the quantum nature of gravity on a certain scale. Apart from the effective quantum gravity interpretation, which employ here, in particular the uni-axial crystal spacetime also describes the propagation of light through a medium, which is where this model actually originated from.

\subsubsection{Lorentzian spacetimes}\label{ssec:LLI}

For a Lorentzian spacetime $(\mathcal{M},g)$, equipped with a metric $g$ with signature $\{1,-1,-1,-1\}$, the point particle dispersion relation is given by $H_x(k) = g^{-1}_x(k,k)$. The local symmetries of the Hamiltonian spacetime $(\mathcal{M},H)$ are given by the Lorentz group, as we proved in Theorem~\ref{thm:LLI}. From Proposition \ref{prop:obstrans} it follows that $\mathcal{G}^\mathrm{or}={\rm SO}^+(3,1)$ and, therefore, Lorentzian spacetimes comply with Observation \ref{ObsGop2} and the technical form of the LMP can be fulfilled. Let us illustrate the operational meaning of $\mathcal{G}^\mathrm{or}={\rm SO}^+(3,1)$ a bit further. Since handedness of co-tetrads and a time orientation can be communicated classically (see above), the group ${\rm SO}^+(3,1)$ quantifies the amount of information that cannot be communicated classically between two observers. The dimension of this group of symmetries is 6.  A co-tetrad has 16 free components. The observer definition conditions in Eq.~\eqref{eq:obsdef} become 
\begin{align}
	g^{-1}(\hat \theta^0,\hat \theta^0) = m^2\quad \textrm{ and }\quad g^{-1}(\hat \theta^A,\hat \theta^0) = 0
\end{align} 
and fix 4 of them. Therefore, we find that the manifold of observer co-tetrads has dimension twelve. Hence, the amount of information that can be gained by comparing the mass shells of two observers is of dimension 6 in the case of a Lorentzian spacetime. The remaining freedom of observers is to orient their spatial co-tetrad $\{\hat \theta^A\}$. The mass-shell encoding of two observers is identical if and only if these observers orient their spatial co-tetrad $\{\hat \theta^A\}$ and $\{\tilde \theta^A\}$ in the same way with respect to the spacetime metric. Thus two observers are in they same equivalence class with respect to their mass-shell encoding $\{\hat \theta^a\}\sim\{\tilde \theta^a\}$ if and only if $g(\hat \theta^A,\hat \theta^B) = g(\tilde \theta^A,\tilde \theta^B)$. This is exactly the information to fix spatial momentum scales and angles. In other words, if all observers fix their spatial co-tetrad components as $g^{-1}(\theta^A,\theta^B)=0$ for $A\neq B$ and $g^{-1}(\theta^A,\theta^A)=1$, there is no additional information to be gained from comparing their respective mass shells; no orientation of the observer's measurement devices is distinguished. This is not the case in a generic Hamiltonian spacetime as we will see in the next example.

\subsubsection{Uniaxial crystal spacetime}
Maxwell's equations inside a uniaxial crystal are given as \cite{Perlick,Fewster:2017mtt}
\begin{align}
\mathcal{F}_{\mu\nu} = \partial_{[ \mu}A_{\nu ]},\quad (2\eta^{\mu[\sigma}\eta^{\rho]\nu} +4 X^{[\sigma} U^{\rho]} X^{[\nu} U^{\mu]})\partial_\rho \mathcal{F}_{\mu\nu}=0\,.
\end{align}
Particle propagation inside a uniaxial crystal is then governed by the Hamiltonian function
\begin{align}
H_x(k)
&= (\eta^{\mu\nu}k_\mu k_\nu)(\eta^{\rho\sigma} - \xi^2 U^\rho U^\sigma + X^\rho X^\sigma )k_\rho k_\sigma = \eta^{-1}(k,k)\zeta^{-1}(k,k)\,.
\end{align}
The vector field $U$ represents the rest-frame of the crystal, while $X$ represents the crystal's optical axis. They satisfy the orthonormality conditions
\begin{align}
\eta(U,U)=1,\quad \eta(X,U) = 0,\quad \eta(X,X) = \xi^2\,.
\end{align}
Therefore, we can always find a coordinate system in which $\eta^{-1}=\mathrm{diag}(1,-1,-1,-1)$, $U=(1,0,0,0)$ and $X=(0,\xi,0,0)$. In this coordinate system, we find $\zeta^{-1}=\mathrm{diag}((1-\xi^2),-(1-\xi^2),-1,-1)$. Hence, the remaining symmetry of the dispersion relation is ${\rm O}(1,1)\times {\rm O}(2)$, which has two real parameters. Since time orientation and handedness can be communicated classically, they reduce the symmetry group to ${\rm SO}^+(1,1)\times {\rm SO}(2)$. Proposition \ref{prop:obstrans} then implies that $\mathcal{G}^\mathrm{or}={\rm SO}^+(1,1)\times {\rm SO}(2)$ and we find that the uniaxial crystal spacetime neither complies with Observation \ref{ObsGop2} nor Observation \ref{ObsGop}; our technical form of the LMP cannot be fulfilled. In the following, we illustrate the meaning of $\mathcal{G}^\mathrm{or}={\rm SO}^+(1,1)\times {\rm SO}(2)$ a bit further. We can conclude that the mass shell can be used to fix the co-tetrad components up to two real parameters (an alternative, more precise derivation of this property is given in Appendix \ref{sec:altinf}). Since the conditions in Eq.~\eqref{eq:obsdef} fix observer co-tetrads already up to 12 free components by the relations
\begin{align}
	\eta^{-1}(\hat\theta^0,\hat\theta^0)\zeta^{-1}(\hat\theta^0,\hat\theta^0) = m^4\quad \textrm{ and }\quad \eta^{-1}(\hat\theta^A,\hat\theta^0)\zeta^{-1}(\hat\theta^0,\hat\theta^0) + \eta^{-1}(\hat\theta^0,\hat\theta^0)\zeta^{-1}(\hat\theta^A,\hat\theta^0) = 0\,,
\end{align} 
the dimension of the manifold of information to be gained from observers comparing mass shell encodings is 10. Again, there remains a freedom of choice in the observer co-tetrad which is the orientation of the spatial co-tetrad components with respect to each other and with respect to the two lightcones. Demanding that two observers have the same mass-shell encoding, $H_{\hat{\theta},x} = H_{\tilde{\theta},x}$, their co-tetrads $\{\hat \theta^a\}$ and $\{\tilde \theta^b\}$ must satisfy
\begin{align}\label{eq:massshelluni}
	\eta^{-1}(\hat \theta^{(a},\hat \theta^b)\zeta^{-1}(\hat \theta^c,\hat \theta^{d)}) = \eta^{-1}(\tilde \theta^{(a},\tilde \theta^b)\zeta^{-1}(\tilde \theta^c,\tilde \theta^{d)})\,.
\end{align}
This is identically satisfied for  $a=b=c=d=0$ and $a=A, b=c=d=0$ by the definition of the observer co-tetrads, but yields additional conditions for any other choice of the indices. Observers who satisfy these equations lie in one equivalence class with respect to their mass shell encoding.

This property shows that fixing spatial co-tetrad would not use all the available information that can be gained by comparing the mass shells in a uniaxial crystal spacetime since \eqref{eq:massshelluni} also imposes additional constraints between the spatial and the $0$-tetrad of different observers in order to have the same mass shell encoding. For example, observers could agree to fix their co-tetrads such that their encodings of the mass shell becomes $(\hat{k}_a {B_{\eta}}^{ab}\hat{k}_b)( \hat{k}_d {{L^{\tau}}^d}_e{B_{\zeta}}^{ef} {L_f}^e \hat{k}_e) = m^4$, where the matrices are given as $B_{\eta}=\mathrm{diag}(1,-1,-1,-1)$ and $B_{\zeta} = \mathrm{diag}((1-\xi^2),-(1-\xi^2),-1,-1)$ and $L$ is an element of the Lorentz group.

\subsubsection{Quantum gravity phenomenology inspired dispersion relations}\label{ssec:kappa}
In quantum gravity phenomenology and in the systematic study of violation of Local Lorentz invariance, inhomogeneous dispersion relations are being explored. A most famous model, which originates in the quantum deformations of the Poincar\'e group \cite{Lukierski:1991pn} and is investigated in the context of non-commutative spacetime \cite{Majid:1994cy}, as well as curved momentum space \cite{Gubitosi:2013rna} geometry models, is the $\kappa$-Poincar\'e dispersion relation.  On flat spacetime it reads
\begin{align}
H_{\kappa}(k) = \frac{4}{\ell^2}\sinh\left(\tfrac{\ell}{2}k_0\right)^2 - e^{\ell k_0}\vec k^2\,,
\end{align}
while on curved spacetime it becomes \cite{Barcaroli:2017gvg}
\begin{align}\label{eq:kappacov}
H_{Z\kappa}(k) = \frac{4}{\ell^2}\sinh\left(\tfrac{\ell}{2}Z(k)\right)^2 - e^{\ell Z(k)}(g^{-1}(k,k) -  Z(k)^2)\,,
\end{align}
where $g$ is a Lorentzian spacetime metric and $Z$ a unit timelike vector field of $g$. Observe that, according to Theorem \ref{thm:LRI}, this Hamiltonian has a ${\rm SO}(3)$ symmetry defined by the metric $h = g - Z\otimes Z$ on the subspace spanned by the co-vectors which are annihilated by $Z$.

To demonstrate what our findings imply for inhomogeneous Hamiltonians, we apply our arguments to \eqref{eq:kappacov}. The equal mass-shell condition $H_{\hat{\theta},x} = H_{\tilde{\theta},x}$ for two observer co-tetrads $\{\hat \theta^a\}$ and $\{\tilde \theta^a\}$ is satisfied if the co-tetrads satisfy sufficient condition
\begin{align}\label{eq:ems}
Z(\hat \theta^a) = Z(\tilde \theta^a)\quad \textrm{ and } \quad g^{-1}(\hat \theta^a,\hat \theta^b) =  g^{-1}(\tilde \theta^a,\tilde \theta^b)\,.
\end{align}
Hence, all observer co-tetrads which satisfy these conditions lie in the same equivalence class and see the same mass shell encoding. 

For the definition of observers there are two possibilities in the context to effective quantum gravity. The first is to argue that observers, as well as the point particles they observe, are subject to the $\kappa$-Poincar\'e dispersion relation. In this case, our definition \eqref{eq:obsdef} of an observer co-tetrad $\{\hat \theta^a\}$ yields
\begin{align}
	\frac{4}{\ell^2}\sinh\left(\tfrac{\ell}{2}Z(\hat \theta^0)\right)^2 - e^{\ell Z(\hat \theta^0)}(g^{-1}(\hat \theta^0,\hat{\theta^0}) -  Z(\hat \theta^0)^2 ) &= m^2\,,\\
	Z_\mu \bar\partial^\mu H(\hat \theta^0) Z(\theta^A) - e^{\ell Z(\hat \theta^0)}(g^{-1}(\hat \theta^0,\hat \theta^A) - Z(\hat{\theta^0})Z(\hat \theta^A)) &= 0\,.
\end{align}
which are additional constraints to the equal mass-shell conditions \eqref{eq:ems}.

The second case is to argue that  observers themselves are not affected by quantum gravity effects and thus not subject to the $\kappa$-Poincar\'e dispersion relation but only to its classical limit $\ell\to 0$, which is given by the Lorentzian metric dispersion relations. In this case the observer definition is given by the classical general relativistic observer co-tetrad conditions
\begin{align}
	g^{-1}(\hat \theta^0,\hat \theta^0) = m^2,\quad \textrm{ and }\quad g^{-1}(\hat \theta^0, \hat \theta^A) = 0\,,
\end{align}
which imply that the $a=0,b=0$ and $a=0,b=B$ equations of the second equal mass-shell constraints \eqref{eq:ems} are identically satisfied for observers. In addition to the orthogonality of their spatial co-tetrad, the observers can use the projection of their co-tetrad onto the vector field $Z$, to align their mass-shells.

\section{Conclusions}

In this work, we have used a novel operational perspective to revisit the question of what distinguishes the description of the spacetime's geometry in terms of a Lorentzian metric from other more general possibilities. Our approach revisits the notion of local inertial laboratories, starting from the observation that both the Einstein equivalence principle and the global Mach principle, which essentially specify what local inertial laboratories are, can be interpreted in generalized geometric structures too. A reason why the EEP and GMP are compatible with many effective spacetime structures is that they are silent on what structures an agent in a local inertial laboratory may maximally exploit for operationally defining the (spatial and temporal) orientation of their reference frame. But it is these local frame orientations and their mutual relations, which encode a lot of information about the local spacetime structure. In particular, in Lorentzian spacetimes, these local orientations of freely falling frames are related by Lorentz transformations. Contemplating operational frame orientations is thus a promising handle for constraining the spacetime environment of a local inertial laboratory that abides by operational formulations of the EEP and GMP.

To this end, we departed from the usual idealizations of local inertial frames in general relativity and took serious that any reference frame is actually a {\it physical} system, which we here describe operationally in terms of a laboratory. A next step is to realize that any experiments, which an agent inside a local laboratory can perform to test its inertial nature, involve matter physics, which ultimately must be described by quantum theory. It is thus natural to characterize local inertial laboratories in terms of the local quantum matter physics inside of it.  Being in free fall, there are no experiments which the agent can carry out to test whether there is any matter outside of their lab and in this sense it can be considered as a system that is operationally isolated from the rest of the matter content of the universe.

Given that we interpret the effective spacetime environment to emerge as a suitable coarse-graining limit of some quantum gravity state, the next natural question to ask is whether the matter physics inside the local inertial lab should or should not be isolated from direct interactions with these  quantum gravitational degrees of freedom too (at the relevant scales of our discussion). While the latter option is clearly an interesting case for quantum gravity phenomenology (on which we also commented), we have mostly focused on the former case in this work, in line with our aim to identify operational statements that single out Lorentzian metrics from within generalized geometries, which we here took to be defined by dispersion relations. 

Specifically, in this light, we have formulated what we called a {\it local Mach principle}, which essentially extends the notion of a local inertial laboratory by requiring that the quantum matter inside of it is {\it also} isolated from net interactions with effective quantum gravitational degrees of freedom (washed out, e.g., through renormalization). Operationally, this means that such a maximally inertial local laboratory must be self-sufficient, so that the agent must generate any reference structures, relative to which they may orient their frame, from the quantum matter physics inside it -- without any `external' help from effective quantum gravity properties. This essentially entails (i) that transformations between different operational frame orientations must emerge purely from quantum matter structures, and (ii) that the local spacetime structure, as seen from a local inertial laboratory, to be completely isotropic (except for a distinction between time and space). However, it does not yet imply by itself a Lorentzian metric structure from within spacetime structures defined by dispersion relations.

In fact, various additional ingredients are necessary to arrive at that conclusion, even after formalizing a version of the LMP in terms of local quantum theory. We note that we had also assumed that, while the local quantum physics will certainly involve, e.g., momentum modes, the agent inside the laboratory may restrict to {\it discrete}, i.e.\ finite-dimensional quantum matter degrees of freedom for the purpose of defining the orientation of their frame. This is certainly a non-trivial restriction, however, as argued, we believe that this assumption can probably be dropped without modifications of the main implications if one works with von Neumann operator algebras. 

At the quantum level, we also had to specify further physical background assumptions to which we subjected the agent in the local laboratory. In particular, we have distinguished the situations in which the {\it eigenvalues} of quantum observables are frame-independent and in which only the {\it rank} of the observables (i.e.\ their zero-eigenvalues) are frame-independent. The former situation led to Observation \ref{ObsGop} and the conclusion that at least three-dimensional rotations must be contained in the set of transformations relating different frame operational frame orientations. The latter situation led to Observation \ref{ObsGop2} and the conclusion that these transformations must, at least, include the orthochronous Lorentz group.

Subsequently, we combined these results with the effective spacetime structures defined by dispersion relations, leading to non-trivial consistency conditions. Firstly, we considered possible observer relations in such spacetime structures, exploiting all operationally accessible structure that observers could, in principle communicate by classical communication (``speakable information''). Specifically, we exploited the mass shell encodings, defined through observer frames (under a choice of momentum encoding map) and the dispersion relation, which here constitutes the ``speakable information''. We showed that if the observers agree on the mass shell encoding, there is no further ``speakable information'' and that the remnant relation between their frames is encoded in a group structure, which defines local and \emph{linear} symmetries of the dispersion relation. 

As we showed, the conjunction with the LMP then implies that these linear and local symmetries of the dispersion relation {\it must} contain at least either the rotations (in the case of Observation \ref{ObsGop}), or the orthochronous Lorentz group (in the case of \ref{ObsGop2}). This was sufficient to show that the dispersion relation {\it must} be a function of a Euclidean metric, in the former case, or even of a Lorentzian metric alone (and momentum vectors), in the latter case. However, the second result by itself does {\it not} imply that the effective spacetime structure is a Lorentzian metric spacetime. For example, the LMP and the second result are consistent with inhomogeneous dispersion relations such as $H_x(k) = \sum_i a_i\ g_x(k,k)^i$, where $g_x(\cdot,\cdot)$ is a Lorentzian metric. As we illustrated in the examples, general relativistic spacetimes do {\it not} feature such dispersion relations (if more than one $a_i\neq0$).

Only if, in addition, we assume the dispersion relation to (a) be homogeneous, (b) be of even degree, and (c) give rise to a well defined initial value problem for the matter described by it \cite{Raetzel:2010je}, does one eventually arrive at Lorentzian metric spacetimes. That is, the entire sequence of steps summarized here is necessary to physically single out Lorentzian metric spacetimes from within spacetime structures defined by dispersion relations. We emphasize, however, that the LMP has played a key role in this sequence.

Our discussion was based on the interpretation that the effective spacetime structures defined by dispersion relations amount to a large-scale coarse-graining limit of some quantum gravity state. However, one could, of course, also interpret some of these spacetime structures alternatively as describing an effective matter structure, the description being technically identical. For example, the uniaxial crystal spacetime \cite{Perlick,Fewster:2017mtt} could also be interpreted as literally a crystal background in which light propagates. The reason we have not entertained such an interpretation here is that it would violate two core assumptions going into the LMP: firstly, that the matter inside of the local inertial laboratory is isolated from any interactions with external matter (the crystal structure, just like the magnetic field $h$ in Example~\ref{ExHeisenberg}, would have to be considered external to the lab) and, secondly, that the agent can, in principle, control all the matter degrees of freedom in their laboratory. 

Finally, we emphasize that a violation of the LMP does {\it not} imply an outright deficiency or operational non-viability of the corresponding dispersion relation. Quite the contrary, in line with our general discussion, the corresponding effective spacetime structure would here be interpreted as a large-scale limit of some quantum gravity state in which the net interactions of the local quantum matter with quantum gravitational degrees of freedom have {\it not} been washed out through renormalization at the relevant laboratory scales. This, in fact, would be the most interesting case for quantum gravity phenomenology.

\section*{Acknowledgements}

PH acknowledges support through a Vienna Center for Quantum Science and Technology Fellowship.
CP acknowledges support by the European Regional Development Fund
through the Center of Excellence TK133 ``The Dark Side of the Universe''.
DR thanks the Humboldt Foundation for funding his research
with their Feodor-Lynen Fellowship.
This research was supported in part by Perimeter Institute for
Theoretical Physics. Research at Perimeter Institute is supported by the
Government of Canada through the Department of Innovation, Science and
Economic Development Canada and by the Province of Ontario through the
Ministry of Research, Innovation and Science. This publication was made
possible through the support of a grant from the John Templeton
Foundation. The opinions expressed in this publication are those of the
authors and do not necessarily reflect the views of the John Templeton
Foundation.

\appendix


\section{Relation between $\cg^{\rm op}$ and $\cg^{\rm or}$}\label{app_gopgor}

We provide the proof of theorem \ref{thm_gopgor}.

\noindent \textbf{Theorem \ref{thm_gopgor}}}. \textit{Let $\cg^{\rm or}_0$ be the connected component at the identity of $\cg^{\rm or}$. If $\mathbb{R}_+$ is a subgroup of $\cg^{\rm or}_0$, we define $\cg^{\rm or}_1$ via $\cg^{\rm or}_0=\cg^{\rm or}_1\times\mathbb{R}_+$, and otherwise $\cg^{\rm or}_1:=\cg^{\rm or}_0$.}
		
\textit{Under the physical background assumptions that lead to Observation~\ref{ObsGop}, we find that either $\cg^{\rm or}_1=\{\mathbf{1}\}$ (the trivial group), or $\cg^{\rm or}_1\supseteq {\rm SO}(3)$. Under the modified physical background assumptions (i), (ii) and (iii) that lead to Observation~\ref{ObsGop2}, and upon invoking Conjecture~\ref{ConjTrans}, we find that either $\cg^{\rm or}_1=\{\mathbf{1}\}$ or $\cg^{\rm or}_1\supseteq {\rm SO}^+(3,1)$, the proper orthochronous Lorentz group.}

	\begin{proof}
		Let us start with the choice of background assumptions that lead to Observation~\ref{ObsGop}. We assume that agents have operational access to the spacetime symmetry group, and that they can apply elements of $\cg^{\rm or}_1$ to their measurement devices (or at least counterfactually understand what would happen if they did so). This way, they have access to the generators of their group as observables. That is, let $\tilde\cg^{\rm or}_1$ be the set of unitaries $U$ with unit determinant such that $U\bullet U^\dagger\in\cg^{\rm or}_1$, which is a subgroup of the special unitary group. Consider its Lie algebra $\mathfrak{g}^{\rm or}$, which is a subspace of antisymmetric operators. The set $i\cdot \mathfrak{g}^{\rm or}$ is a set of observables that is operationally distinguished for any observer: choosing any operator basis, agents can tomographically determine the action of the group $\cg^{\rm or}_1$ on their observables, and thus determine a description of $i\cdot\mathfrak{g}^{\rm or}$ (and, if they want, they can construct measurement devices that measure the corresponding observables). The resulting set of observables is independent of the initial choice of basis, in the sense that agents with different choices of basis would be led to the construction of the same family of devices.
		
		Now suppose there exists $V\bullet V^\dagger\in\cg^{\rm op}$ such that $V\mathfrak{g}^{\rm or}V^\dagger \neq \mathfrak{g}^{\rm or}$. Then $V\mathfrak{g}^{\rm or}V^\dagger$ and $\mathfrak{g}^{\rm or}$ are two equivalent subspaces of observables which can be operationally distinguished: one contains all the generators of spacetime symmetries and the other one doesn't. This contradicts the LMP. In other words, the LMP implies that $V\mathfrak{g}^{\rm or}V^\dagger=\mathfrak{g}^{\rm or}$ for all $V\bullet V\in\cg^{\rm op}$, or $e^{tG}Xe^{-tG}\in\mathfrak{g}^{\rm or}$ for all $X\in\mathfrak{g}^{\rm or}$ and $G\in\mathfrak{g}^{\rm op}$, where $\mathfrak{g}^{\rm op}=\bigoplus_{i=1}^n {\rm su}(n_i)$ is the Lie algebra of $\cg^{\rm op}$ (cf.\ Observation~\ref{ObsGop}), with $n_i\geq 2$. Differentiating at $t=0$ implies that $[G,X]\in\mathfrak{g}^{\rm or}$, hence $\mathfrak{g}^{\rm or}$ is an \emph{ideal} of $\mathfrak{g}^{\rm op}$. Taking into account the direct sum decomposition of $\mathfrak{g}^{\rm op}$ into simple Lie algebras ${\rm su}(n_i)$, it follows that there exists $m\in [0,n]$ and $1\leq i_1<i_2<\ldots<i_m\leq n$ such that $\mathfrak{g}^{\rm or}=\bigoplus_{j=1}^m {\rm su}(n_{i_j})$ \cite{bourbaki1975elements}. If $m=0$ then $\mathfrak{g}^{\rm or}=\{0\}$, and hence $\tilde\cg^{\rm or}_1$ and thus $\cg^{\rm or}_1$ is trivial. This describes a situation in which there simply are no spacetime symmetries (other than perhaps rescalings $A\mapsto \lambda A$ that we have factored out from the start). Such a situation is perfectly compatible with the LMP. On the other hand, if $m\geq 1$, $\mathfrak{g}^{\rm or}$ contains at least one summand ${\rm su}(n_i)$ with $n_i\geq 2$. In this case, $\tilde \cg^{\rm or}_1$ contains at least ${\rm SU}(2)$, and, by conjugation, $\cg^{\rm or}_1$ contains at least ${\rm PSU}(2)={\rm SO}(3)$.
		
		Now consider the background assumptions (i), (ii) and (iii) that have led us to Observation~\ref{ObsGop2}. In this case, the elements of $\mathfrak{g}^{\rm or}$ do not transform as observables (they are also neither all symmetric nor all antisymmetric), thus the proof strategy of above cannot be generalized. (In fact, even the notion of operator multiplication and thus of Lie bracket becomes basis-dependent.) We will use a different argumentation. First, note that every element of $\cg^{\rm or}_1$ is a map $A\mapsto Y_1\otimes\ldots \otimes Y_n A Y_1^\dagger\otimes\ldots\otimes Y_n^\dagger$, where $Y_i\in{\rm SL}(n_i,\mathbb{C})$. Thus, we can simply pick one of the $\mathcal{A}_i$ for $1\leq i \leq n$ and study the restriction of the action of $\cg^{\rm or}_1$ to $\mathcal{A}_i$. We will now show that this action is either trivial, or corresponds to the full projective special linear group ${\rm PSL}(n_i,\mathbb{C})$. Only considering one of the $\mathcal{A}_i$ is equivalent to assuming $n=1$, which we will thus do from now on.
		
		Analogously to further above, we can define the group $\tilde\cg^{\rm or}_1$ as the set of all $X\in{\rm SL}(n_1,\mathbb{C})$ such that $X\bullet X^\dagger\in\cg^{\rm or}_1$. We define the group $\tilde\cg^{\rm op}_1$ analogously; note that $\cg^{\rm op}_1$ is the projective special linear group ${\rm PSL}(n_1,\mathbb{C})$, hence $\tilde\cg^{\rm op}_1={\rm SL}(n_1,\mathbb{C})$. A \emph{$\cg^{\rm or}_1$-orbit} is a non-empty set $\{XAX^\dagger\,\,|\,\, X\in\tilde \cg^{\rm or}_1\}$, with $A=A^\dagger$ an arbitrary observable. The LMP implies the following: if $S$ is a $\cg^{\rm or}_1$-orbit, then so is $YSY^\dagger$ for every $Y\in\tilde\cg^{\rm op}_1$. Now, for every pair of non-negative rank-one observables $P,Q$, write $P\sim Q$ if $P$ and $Q$ are elements of a common $\cg^{\rm or}_1$-orbit, i.e.\ if there exists $X\in\tilde\cg^{\rm or}_1$ with $XP X^\dagger = Q$.
		
		Now we make the following claim: \emph{If $P$, $P'$, $Q$, $Q'$ are non-negative rank-one observables, and if $P\sim P'$ and $P\sim Q$ and if there exists $Y\in{\rm SL}(n_1,\mathbb{C})$ with $YPY^\dagger=P'$ and $YQY^\dagger=Q'$ then $P\sim Q'$.} To prove this claim, note that the assumptions imply that $P$, $P'$ and $Q$ are all on the same $\cg^{\rm or}_1$-orbit, and $YPY^\dagger=P'$ means that $Y\bullet Y^\dagger$ maps this orbit onto itself. Hence $Q$ and $Q'$ also lie on this very same orbit, and consequently $P\sim Q'$.
		
		This insight has an interesting consequence. Let $S$ be any $\cg^{\rm or}_1$-orbit, i.e.\ an equivalence class of non-negative rank-one observables with respect to the relation $\sim$. Let $Y\in{\rm SL}(n_1,\mathbb{C})$, and suppose that there exists some $Q\in S$ such that $YQY^\dagger\not\in S$. Then there \emph{cannot} be two elements $P,P'\in S$ with $YPY^\dagger=P'$, otherwise one would obtain a contradiction to the insight above. But this implies that $Y\bullet Y^\dagger$ maps \emph{all} elements of $S$ outside of $S$. On the other hand, if $Z\in{\rm SL}(n_1,\mathbb{C})$ maps at least one element of $S$ inside of $S$, then $ZS Z^\dagger\subseteq S$, and since orbits are disjoint we must actually have $ZSZ^\dagger=S$. In summary, we conclude that to $S$ there corresponds a group $\cg_S\subseteq \cg^{\rm op}_1$ such that
		\begin{itemize}
			\item $Y\bullet Y^\dagger\in\cg_S\Leftrightarrow YSY^\dagger=S$,
			\item $Y\bullet Y^\dagger\in\cg^{\rm op}_1\setminus \cg_S \Leftrightarrow YSY^\dagger\cap S =\emptyset$.
		\end{itemize}
		In particular, $\cg^{\rm or}_1\subseteq \cg_S$. Define the group $\tilde\cg_S:=\{Y\in{\rm SL}(n_1,\mathbb{C})\,\,|\,\, Y\bullet Y^\dagger\in\cg_S\}$. Now here is yet another claim: \emph{If $\psi,\varphi\in\mathbb{C}^{n_1}\setminus\{0\}$ are any two non-zero vectors, then $|\psi\rangle\langle\psi|\sim |\varphi\rangle\langle\varphi|$ if and only if there exists some $X\in\tilde\cg_S$ such that $X\psi=\varphi$, where $S$ is the $\cg^{\rm or}_1$-orbit containing $|\psi\rangle\langle\psi|$.} One direction of this claim is easy to show: if $X\psi=\varphi$ for $X\in\tilde\cg_S$ then $X|\psi\rangle\langle\psi|X^\dagger=|\varphi\rangle\langle\varphi|$ for $X\bullet X^\dagger\in\cg_S$, and so $|\psi\rangle\langle\psi|\in S$ implies that $|\varphi\rangle\langle\varphi|\in S$ too -- that is, $|\psi\rangle\langle\psi|\sim|\varphi\rangle\langle\varphi|$. Conversely, if $|\psi\rangle\langle\psi|\sim|\varphi\rangle\langle\varphi|$ then there is some $Y\in\tilde\cg^{\rm or}_1$ with $Y|\psi\rangle\langle\psi| Y^\dagger=|\varphi\rangle\langle\varphi|$, i.e.\ $Y\psi=e^{i\theta}\varphi$ for some $\theta\in\mathbb{R}$. Since ${\rm SL}(n_1,\mathbb{C})$ is transitive on $\mathbb{C}^{n_1}\setminus\{0\}$ \cite{geatti2012some}. there exists $Z\in{\rm SL}(n_1,\mathbb{C})$ with $Z\varphi=e^{-i\theta}\varphi$. Since $Z|\varphi\rangle\langle\varphi|Z^\dagger=|\varphi\rangle\langle\varphi|$ and $Z\bullet Z^\dagger\in\cg^{\rm op}_1$, we must have $Z\bullet Z^\dagger\in\cg_S$. But $Y\bullet Y^\dagger\in\cg^{\rm or}_1\subseteq \cg_S$, hence $(ZY)\bullet(ZY)^\dagger\in\cg_S$ and $ZY\psi=\varphi$. This proves our intermediate claim.
		
		Suppose that \emph{all} $\cg^{\rm or}_1$-orbits of non-negative rank-one observables contain exactly one element. Then $\cg^{\rm or}_1$ must be the trivial group, which is the first possibility listed in the statement of the theorem. On the other hand, suppose there exists some $\cg^{\rm or}_1$-orbit $S$ that contains more than one element. If all elements of $S$ are scalar multiples of a single element $|\psi\rangle\langle\psi|$ (such that $S$ is the subset of a single ray), then the same is true for $YSY^\dagger$, for all $Y\in{\rm SL}(n_1,\mathbb{C})$. Again using the transitivity of that group on $\mathbb{C}^{n_1}\setminus\{0\}$, it follows every $X\in\tilde \cg^{\rm or}_1$ acts as $X\psi=\lambda_{\psi}\psi$ for every $\psi\in\mathbb{C}^{n_1}\setminus\{0\}$, where $\lambda_{\psi}\in\mathbb{C}$, and this is only possible if $X$ is a multiple of the identity, which is a contradiction -- we have factored out all multiples of the identity from the start. Hence there must be at least two vectors $\psi,\varphi\in\mathbb{C}^{n_1}\setminus\{0\}$ which are linearly independent, and $|\psi\rangle\langle\psi|,|\varphi\rangle\langle\varphi|\in S$. Define the stabilizer group $\tilde\cg_\psi:=\{X\in{\rm SL}(n_1,\mathbb{C})\,\,|\,\, X\psi=\psi\}$ and analogously define $\tilde\cg_\varphi$, and let $\tilde\cg_{\psi,\varphi}$ be the smallest group containing both. It follows that
		\[
		\tilde\cg_\psi\subsetneq \tilde\cg_{\psi,\varphi}\subseteq \tilde \cg_S \subseteq {\rm SL}(n_1,\mathbb{C}).
		\]
		But the stabilizer subgroup $\tilde\cg_\psi$ is a maximal subgroup of ${\rm SL}(n_1,\mathbb{C})$ \cite{king1981some} hence $\tilde\cg_{\psi,\varphi}=\tilde\cg_S={\rm SL}(n_1,\mathbb{C})$. Again using the transitivity of this group, this implies that $S$ in fact contains \emph{all} non-negative rank-one observables; that is, $\tilde\cg^{\rm or}_1$ acts transitively by conjugation on the full set of non-negative rank-one observables.
		
		Now we use Conjecture~\ref{ConjTrans}: it tells us that $\mathcal{\tilde G}^{\rm or}_1={\rm SL}(n_1,\mathbb{C})$ or (if $n_1$ is even) $\mathcal{\tilde G}^{\rm or}_1={\rm Sp}(n_1,\mathbb{C})$. Clearly, for all $n_1\geq 2$, ${\rm SL}(2,\mathbb{C})$ is a subgroup of ${\rm SL}(n_1,\mathbb{C})$, and for $n_1=2$ we have ${\rm SL}(2,\mathbb{C})={\rm Sp}(2,\mathbb{C})$. Furthermore, for $k\geq 2$, ${\rm SL}(k,\mathbb{C})$ is a subgroup of ${\rm Sp}(2k,\mathbb{C})$, which can be seen by mapping $A\in{\rm SL}(k,\mathbb{C})$ to $\left(\begin{array}{cc} A & 0 \\ 0 & (A^\top)^{-1} \end{array}\right)\in{\rm Sp}(2k,\mathbb{C})$. All in all,we obtain that ${\rm SL}(2,\mathbb{C})$ is a subgroup of $\mathcal{\tilde G}^{\rm or}_1$, and thus, by conjugation, ${\rm PSL}(2,\mathbb{C})={\rm SO}^+(3,1)$ is a subgroup of $\cg^{\rm or}_1$.
	\end{proof}

\section{Proof Theorem: Local Lorentz invariant dispersion relations}\label{app:thmLLI}
In Sect.~\ref{sec:DispHam} we stated the Theorems \ref{thm:LLI} and \ref{thm:LRI}. Here we present their proofs. \\

\noindent \textbf{Theorem \ref{thm:LLI}} (Local Lorentz invariant dispersion relations).
\textit{Consider a Hamiltonian spacetime $(M,H)$ and let $g$ be some Lorentzian spacetime metric. The generators of local Lorentz transformations on the cotangent spaces of spacetime $M^{\mu\nu} = g^{\mu\sigma}k_\sigma \bar{\partial}^\nu - g^{\nu\sigma}k_\sigma \bar{\partial}^\mu$ generate local symmetries of $H$ if and only if $H_x(k) = h_x(w(k))$, where $w(k) = g^{\mu\nu}(x)k_\mu k_\nu$ and $h_x(w)$ is a function in one variable only.}\\

\noindent {\textbf {Proof:}} Assume $H_x(k) = H_x(w(k))$ with $w(k) = g^{\mu\nu}(x)k_\mu k_\nu$. Then, 
\begin{align}
M^{\mu\nu}(H_x) = (g^{\mu\sigma}k_\sigma \bar{\partial}^\mu w - g^{\nu\sigma}k_\sigma \bar{\partial}^\nu w)\partial^w H_x = 0\,.
\end{align} 
The other way around, assume that $M^{\mu\nu}(H_x) = 0$, then $k_\nu M^{\mu\nu}H_x = 0$ implies
\begin{align}\label{eq:thm2prf1}
\bar{\partial}^\mu H_x = g^{\mu\sigma} k_\sigma \frac{k_\rho \bar{\partial}^\rho H_x}{w} = k^\mu \frac{Q}{w},\quad \textrm{where } g^{\mu\sigma} k_\sigma = k^\mu,\ Q = k_\rho \bar{\partial}^\rho H_x\,.
\end{align}
Now introduce new coordinates $\tilde k_0(k) := w(k),\ \tilde k_A(k) := k_A,\ A=1,2,3$ and use them to study the consequences of \eqref{eq:thm2prf1}. For $\mu=0$ we conclude from 
\begin{align}
\bar{\partial}^0 H_x = k^0 \frac{Q}{w}\quad &\textrm{and}\quad \bar{\partial}^0 H_x = (\bar{\partial}^0 \tilde k_\nu) \tilde {\bar \partial}^\nu H_x = 2 k^0 \tilde {\bar \partial}^0 H_x
\end{align}
that $2 w \tilde {\bar \partial}^0 H_x = Q$. For $\mu=A$ we use 
\begin{align}
\bar{\partial}^A H_x = k^A \frac{Q}{w}\quad &\textrm{and}\quad \bar{\partial}^A H_x =(\bar{\partial}^A \tilde k_\nu) \tilde {\bar \partial}^\nu H_x = 2 k^A \tilde {\bar \partial}^0 H_x + \tilde {\bar \partial}^A H_x = k^A \frac{Q}{w} + \tilde {\bar \partial}^A H_x
\end{align}
to finally conclude that $\tilde {\bar \partial}^A H_x = 0$. Hence $H_x(k(\tilde k)) = h_x(\tilde k_0) = h_x(w)$ is just a function of the variable $w$. $\square$\\

\noindent \textbf{Theorem \ref{thm:LRI}} (Rotational invariant dispersion relations).
\textit{Consider a Hamiltonian spacetime $(M,H)$ and let $\Sigma_x\subset T_xM$ be a three dimensional sub-vector space of $T^*_xM$ equipped with a positive/negative definite scalar product $s$ and coordinates $\{p_A = A^\mu{}_A(x) k_\mu\}_{A=1}^3$. The generators of orthogonal transformations in $(\Sigma_x,s)$ are $M^{AB} = s^{AC}p_C \bar{\partial}^B - s^{BC}p_C \bar{\partial}^A$. They generate local symmetries of $H$ if and only if $H_x(k) = r_x(p_0,v(p))$, where $p_0$ completes the sub-vector space coordinates $\{p_A\}_{A=1}^3$ to coordinates on $T^*_xM$, $v(p) = s^{AB}(x)p_Ap_B$ and $r_x(p_0,v)$ is a function in two variables only.}\\

\noindent {\textbf {Proof:}} The proof works analogue to one of Theorem \ref{thm:LLI}. Assume $H_x(k) = H_x(p_0,v)$ with $v(p) = s^{AB}(x)p_Ap_B$, then 
\begin{align}
	M^{AB}(H) = (s^{AC}p_C \bar{\partial}^B v- s^{BC}p_C \bar{\partial}^A v)\partial^v H_x = 0\,.
\end{align}
The other way around, assume that $M^{AB}(H_x) = 0$, then $p_B M^{AB}H_x = 0$ implies
\begin{align}\label{eq:thm3prf1}
	\bar{\partial}^A H_x = s^{s^{AC}p_C}\frac{p_B\bar{\partial}^BH_x}{v} = p^A \frac{R}{v},\quad \textrm{where } s^{AC} p_C = p^A,\ R = p_B\bar{\partial}^BH_x\,.
\end{align}
Again introduce new coordinates $\tilde p_1(k) := v(p),\ \tilde p_\mathfrak{a} := p_\mathfrak{a},\ \mathfrak{a}=2,3$ and use them to study the consequences of \eqref{eq:thm3prf1}. For $A=1$ we conclude from 
\begin{align}
\bar{\partial}^1 H_x = p^1 \frac{R}{v}\quad &\textrm{and}\quad \bar{\partial}^1 H_x = (\bar{\partial}^1 \tilde p_A) \tilde {\bar \partial}^A H_x = 2 p^1 \tilde {\bar \partial}^1 H_x
\end{align}
that $2 v \tilde {\bar \partial}^1 H_x = R$. For $A=\mathfrak{a}$ we use 
\begin{align}
\bar{\partial}^\mathfrak{a} H_x = p^\mathfrak{a} \frac{R}{v}\quad &\textrm{and}\quad \bar{\partial}^\mathfrak{a} H_x =(\bar{\partial}^\mathfrak{a} \tilde p_A) \tilde {\bar \partial}^A H_x = 2 p^A \tilde {\bar \partial}^1 H_x + \tilde {\bar \partial}^\mathfrak{a} H_x = p^\mathfrak{a} \frac{R}{v} + \tilde {\bar \partial}^\mathfrak{a} H_x
\end{align}
to finally conclude that $\tilde {\bar \partial}^\mathfrak{a} H_x = 0$. Hence $H_x(k) = r_x(p_0, v(p))$ is just a function of the variables $p_0$ and $v$. $\square$
  
\section{Observer dependence of observer transformations\label{sec:obsdep}}
Let us consider infinitesimal frame transformations $\tilde{\theta}^a = \hat \theta^a + \delta \theta^a$. From the observer frame conditions in Eq. \eqref{eq:obsdef} we obtain the two conditions
\begin{align}
\delta \theta^0{}_{\mu}\bar{\partial}^\mu H_x(\hat \theta^0) = 0
\quad \mathrm{and}\quad
\delta \theta^A{}_\mu \bar\partial^\mu H_x(\hat \theta^0) + \delta \theta^0{}_\mu \theta^A{}_\nu \bar\partial^\mu \bar\partial^\nu H_x(\hat \theta^0) = 0 
\end{align}
From the first condition, we see that $\delta \theta^0$ lies in the annihilator of $e_0{}^\mu = \bar \partial^\mu H_x(\hat \theta^0)$. Therefore, $\delta \theta^0$ is spatial with respect to the tetrad $\{\theta^a\}$ and we can express $\delta \theta^0$ by choosing a three-co-vector $\vec s$ and writing $\delta \theta^0 = s_A \theta^A$. Then, from the second condition, we deduce that, in general, $\delta \theta^A$ has to depend on $\theta^0$ since $\bar{\partial}^\mu\bar{\partial}^\nu H_x(\theta^0)$ depends on $\theta^0$ for a generic dispersion relation.

\section{Proof of Proposition on observers' identical mass shell encodings} \label{app:propObsMS}

In Sec.~\ref{sec:ObsDispRleations} we stated Prop~\ref{prop:ObsMS}, which we prove here.\\

\noindent \textbf{Proposition 1} (Identical mass-shell encodings). 
\textit{Let $\{\hat \theta^a\}_{a=0}^3$ and $\{\tilde \theta^a\}_{a=0}^3$ be observer tetrads on the Hamiltonian spacetime $(M,H)$, which are related by the observer transformation matrix $\Lambda^a{}_b(\hat \theta,\tilde{\theta})$. The observer agree on their encodings of the dispersion relation $H_{\hat \theta, x}$ and $H_{\tilde \theta, x}$ via the encoding map $f_{\hat \theta}$ and $f_{\tilde \theta}$ if and only if the observer transformation is a local \emph{linear} symmetry of the Hamiltonian.}\\

\noindent{\textit{Proof:}} Assume $H_{\hat \theta, x}(\hat k) = H_{\tilde \theta, x}(\hat k)$. Then, 
\begin{align}\label{eq:proofprop1}
H_x(k) = H_x(\hat \theta^a \hat k_a) = H_x(\tilde \theta^a \hat k_a) = H_x((\Lambda^{-1})^a{}_b \hat \theta^b \hat k_b) = H_x((\Lambda^{-1})^a{}_b \hat \theta^b \hat e_a{}^\mu k_\mu) = H_x(\Psi_{\Lambda}(k))\,.
\end{align} 
Hence the observer transformation induced local map $\Psi_\Lambda$ from $T^*_xM$ to $T^*_xM$, defined by $\Psi_{\Lambda}{}^\mu{}_\nu = (\Lambda^{-1})^a{}_b \hat \theta^b{}_\nu \hat e_a{}^\mu $ is a local symmetry of $H$ and, in particular, its a linear map acting on $k$: $\Psi_\Lambda(k)_\nu = \Psi_\Lambda{}^\mu{}_\nu k_\mu$. The other way around, if $\Psi_{\Lambda}{}^\mu{}_\nu = (\Lambda^{-1})^a{}_b \hat \theta^b{}_\nu \hat e_a{}^\mu $ defines a local symmetry, then $H_x(\Psi_{\Lambda}(k)) = H_x(k)$ and the analogue manipulations as in \eqref{eq:proofprop1} show that $H_{\hat \theta, x}(\hat k) = H_{\tilde \theta, x}(\hat k)$. $\square$

\section{Proof of Proposition 2\label{sec:proofprop2}}
Here we proof Prop.~\ref{prop:obstrans} of Sec.~\ref{sec:ObsDispRleations}.\\

\textbf{Proposition 2} (Observer transformations and local and linear symmetries).
\textit{Let $\cg^{\rm dis}$ be the group of local and linear symmetries of $H_x$. For each co-frame $\{\theta^a\}_{a=0}^3$, the map $I_\theta : \cg^{\rm dis} \rightarrow \mathfrak{G}_\theta \subset GL(4)$ with $I_\theta (\Psi)^a{}_b = \Lambda^a{}_b(\Psi(\theta),\theta)$ for all $\Psi\in \cg^{\rm dis}$ is a group isomorphism. Furthermore, $\mathfrak{G}_\theta$ is equivalent to the set of all observer transformations between observer frames in the same equivalence class as the frame~$\{\theta^a\}_{a=0}^3$. }\\

\noindent{\textit{Proof:}} Let us consider a coordinate system around $x$ and the corresponding basis at $x$, $\{dx^\mu,\partial_{\mu}\}$. In this basis, a co-frame and its respective dual frame are given by the expressions $\theta^a{}_\mu=\theta^a(\partial_\mu)$ and $e_a{}^\mu = dx^\mu(e_a)$, which satisfy $\theta^a{}_\nu e_a{}^\mathcal{\mu} = \delta^\mu_\nu$ since  the duality relation $\theta^a{}(e_b) = \delta^a_b$ implies $e_a{}^\sigma \theta^a{}_\mu e_b{}^\mu = e_b{}^\sigma$.

Furthermore, $I_\theta (\Psi)^a{}_b = \theta^a{}_\mu \Psi^\mu{}_\nu e_b{}^\nu$. Therefore, the inverse of $I_\theta$ is the map $I_\theta^{-1}(\Lambda)^\mu{}_\nu= e_a{}^\mu \Lambda^a{}_b  \theta^b{}_\nu$. For the product of two elements of $\mathfrak{G}_\theta$, we find 
\begin{align*}
	I_\theta(\Psi)^a{}_b I_\theta(\bar\Psi)^b{}_c = Q^a{}_b(\Psi(\theta),\theta) Q^b{}_c(\bar\Psi(\theta),\theta) = \theta^a{}_\mu \Psi^\mu{}_\nu e_b{}^\nu  \theta^b{}_\rho \bar\Psi^\rho{}_\sigma e_c{}^\sigma = \theta^a{}_\mu \Psi^\mu{}_\rho \bar\Psi^\rho{}_\sigma e_c{}^\sigma = I_\theta(\Psi(\bar\Psi))^a{}_c\,. 
\end{align*}

Next, let us consider two observer co-frames $\{\tilde\theta^a\}$ and $\{\hat\theta^a\}$ with $\{\tilde\theta^a \} \sim \{ \theta^a\} \sim \{\hat\theta^a\}$, i.e.\ $\exists$ $\tilde\Psi,\,\hat\Psi\in \cg^{\rm dis}$ such that $\tilde\theta^a=\tilde\Psi(\theta^a)$ and $\hat\theta^a = \hat\Psi(\theta^a)$. Since $\hat\Psi(\theta^a)_\mu=\theta^a{}_\nu\hat\Psi^\nu{}_\mu$, we have $\hat e_a{}^\mu= (\hat\Psi^{-1})^\mu{}_\nu e_a{}^\nu $, where $(\hat\Psi^{-1})^\nu{}_\mu\hat\Psi^\mu{}_\rho=\delta^\nu_\rho$. Then, $\Lambda^a{}_b(\tilde\theta,\hat\theta)=\tilde\theta^a(\hat{e}_b)=\theta^a{}_\nu\tilde\Psi^\nu{}_\mu (\hat\Psi^{-1})^\mu{}_\rho e_a{}^\rho $. Since $\tilde\Psi(\hat\Psi^{-1})\in \cg^{\rm dis}$, we find that the co-frame $\theta^a{}_\nu\tilde\Psi^\nu{}_\mu (\hat\Psi^{-1})^\mu{}_\rho$ is in the same equivalence class as $\theta$ and  $\Lambda^a{}_b(\tilde\theta,\hat\theta) \in \mathfrak{G}_\theta$. $\square$

\section{Uniaxial crystal spacetime - information through mass shell comparison \label{sec:altinf}}

Alternatively, we can derive the dimension of the manifolds of information that can be gained from comparing mass shells in an uniaxial crystal spacetime by investigating the tensor 
\begin{align}
G^{\mu\nu\rho\sigma}:=\frac{1}{4!}\partial_{y_\mu}\partial_{y_\nu}\partial_{y_\rho}\partial_{y_\sigma}D_x(y)^2=\eta^{(\mu\nu}(\eta^{\rho\sigma)} - \xi^2 U^\rho U^{\sigma)} + X^\rho X^{\sigma)} )
\end{align}
and its encoding $G(\theta)^{\mathcal{A}\mathcal{B}\mathcal{C}\mathcal{D}}:=G^{\mu\nu\rho\sigma}\theta^\mathcal{A}_\mu \theta^\mathcal{B}_\nu \theta^\mathcal{C}_\rho \theta^\mathcal{D}_\sigma$. The encoding $G(\theta)^{\mathcal{A}\mathcal{B}\mathcal{C}\mathcal{D}}$ represents the maximal ammount of information an observer can obtain about the spacetime structure in case of the uniaxial crystal spacetime.
$G(\theta)^{\mathcal{A}\mathcal{B}\mathcal{C}\mathcal{D}}$ induces a map from the 16 components of the co-frame to the $35$ components of a general completely symmetrized tensor of degree 4 in 4 dimensions. By defining a map $ind:[1,35]\subset \mathbb{Z} \rightarrow [1,4]^4\subset\mathbb{Z}^4$ such that the vector with components $V_G(\theta)^{J}:=G(\theta)^{ind(J)}$ contains all 35 independent components of $G^{\mathcal{A}\mathcal{B}\mathcal{C}\mathcal{D}}$ and defining the vector $v_\theta=(\theta^0_0,...,\theta^3_3)\in \mathbb{R}^{16}$ containing all $16$ components of $\theta$, we obtain the nonlinear map $Enc_G:\mathbb{R}^{16}\rightarrow \mathbb{R}^{35}$ defined as $Enc_G(v_\theta)=V_G(\theta)$. Analyses performed with Wolfram Mathematica revealed that the Jacobian matrix of $Enc_G$ has rank 14, which implies that only 14 of 16 components of the co-frame can be fixed 
by comparing two encodings of the dispersion relation. The remaining uncertainty amounts to two real numbers, which coincides with the result we obtain from the analysis of the remaining symmetries of the dispersion relation. 
\bibliographystyle{utphys}
\bibliography{LMPC}

\providecommand{\href}[2]{#2}\begingroup\raggedright\begin{thebibliography}{10}

\bibitem{Riess:1998cb}
{\bfseries Supernova Search Team} Collaboration, A.~G. Riess {\em et~al.},
  ``{Observational evidence from supernovae for an accelerating universe and a
  cosmological constant},'' \href{http://dx.doi.org/10.1086/300499}{{\em
  Astron. J.} {\bfseries 116} (1998) 1009--1038},
\href{http://arxiv.org/abs/astro-ph/9805201}{{\ttfamily arXiv:astro-ph/9805201
  [astro-ph]}}.

\bibitem{Peebles:2002gy}
P.~J.~E. Peebles and B.~Ratra, ``{The Cosmological constant and dark energy},''
  \href{http://dx.doi.org/10.1103/RevModPhys.75.559}{{\em Rev. Mod. Phys.}
  {\bfseries 75} (2003) 559--606},
  \href{http://arxiv.org/abs/astro-ph/0207347}{{\ttfamily
  arXiv:astro-ph/0207347 [astro-ph]}}.
[,592(2002)].

\bibitem{Corbelli:1999af}
E.~Corbelli and P.~Salucci, ``{The Extended Rotation Curve and the Dark Matter
  Halo of M33},''
  \href{http://dx.doi.org/10.1046/j.1365-8711.2000.03075.x}{{\em Mon. Not. Roy.
  Astron. Soc.} {\bfseries 311} (2000) 441--447},
\href{http://arxiv.org/abs/astro-ph/9909252}{{\ttfamily arXiv:astro-ph/9909252
  [astro-ph]}}.

\bibitem{Clowe:2006eq}
D.~Clowe, M.~Bradac, A.~H. Gonzalez, M.~Markevitch, S.~W. Randall, C.~Jones,
  and D.~Zaritsky, ``{A direct empirical proof of the existence of dark
  matter},'' \href{http://dx.doi.org/10.1086/508162}{{\em Astrophys. J.}
  {\bfseries 648} (2006) L109--L113},
\href{http://arxiv.org/abs/astro-ph/0608407}{{\ttfamily arXiv:astro-ph/0608407
  [astro-ph]}}.

\bibitem{AmelinoCamelia:2008qg}
G.~Amelino-Camelia, ``{Quantum-Spacetime Phenomenology},''
  \href{http://dx.doi.org/10.12942/lrr-2013-5}{{\em Living Rev. Rel.}
  {\bfseries 16} (2013) 5},
\href{http://arxiv.org/abs/0806.0339}{{\ttfamily arXiv:0806.0339 [gr-qc]}}.

\bibitem{Hossenfelder:2012jw}
S.~Hossenfelder, ``{Minimal Length Scale Scenarios for Quantum Gravity},''
  \href{http://dx.doi.org/10.12942/lrr-2013-2}{{\em Living Rev. Rel.}
  {\bfseries 16} (2013) 2},
\href{http://arxiv.org/abs/1203.6191}{{\ttfamily arXiv:1203.6191 [gr-qc]}}.

\bibitem{Liberati:2013xla}
S.~Liberati, ``{Tests of Lorentz invariance: a 2013 update},''
  \href{http://dx.doi.org/10.1088/0264-9381/30/13/133001}{{\em Class. Quant.
  Grav.} {\bfseries 30} (2013) 133001},
\href{http://arxiv.org/abs/1304.5795}{{\ttfamily arXiv:1304.5795 [gr-qc]}}.

\bibitem{Mattingly:2005re}
D.~Mattingly, ``{Modern tests of Lorentz invariance},''
  \href{http://dx.doi.org/10.12942/lrr-2005-5}{{\em Living Rev. Rel.}
  {\bfseries 8} (2005) 5},
\href{http://arxiv.org/abs/gr-qc/0502097}{{\ttfamily arXiv:gr-qc/0502097
  [gr-qc]}}.

\bibitem{Heisenberg:2018mxx}
L.~Heisenberg, R.~Kase, and S.~Tsujikawa, ``{Cosmology in scalar-vector-tensor
  theories},'' \href{http://dx.doi.org/10.1103/PhysRevD.98.024038}{{\em Phys.
  Rev.} {\bfseries D98} no.~2, (2018) 024038},
\href{http://arxiv.org/abs/1805.01066}{{\ttfamily arXiv:1805.01066 [gr-qc]}}.

\bibitem{Sotiriou:2008rp}
T.~P. Sotiriou and V.~Faraoni, ``{f(R) Theories Of Gravity},''
  \href{http://dx.doi.org/10.1103/RevModPhys.82.451}{{\em Rev. Mod. Phys.}
  {\bfseries 82} (2010) 451--497},
\href{http://arxiv.org/abs/0805.1726}{{\ttfamily arXiv:0805.1726 [gr-qc]}}.

\bibitem{Capozziello:2011et}
S.~Capozziello and M.~De~Laurentis, ``{Extended Theories of Gravity},''
  \href{http://dx.doi.org/10.1016/j.physrep.2011.09.003}{{\em Phys. Rept.}
  {\bfseries 509} (2011) 167--321},
\href{http://arxiv.org/abs/1108.6266}{{\ttfamily arXiv:1108.6266 [gr-qc]}}.

\bibitem{Lovelock:1971yv}
D.~Lovelock, ``{The Einstein tensor and its generalizations},''
\href{http://dx.doi.org/10.1063/1.1665613}{{\em J. Math. Phys.} {\bfseries 12}
  (1971) 498--501}.

\bibitem{deRham:2014zqa}
C.~de~Rham, ``{Massive Gravity},''
  \href{http://dx.doi.org/10.12942/lrr-2014-7}{{\em Living Rev. Rel.}
  {\bfseries 17} (2014) 7},
\href{http://arxiv.org/abs/1401.4173}{{\ttfamily arXiv:1401.4173 [hep-th]}}.

\bibitem{Clifton:2011jh}
T.~Clifton, P.~G. Ferreira, A.~Padilla, and C.~Skordis, ``{Modified Gravity and
  Cosmology},'' \href{http://dx.doi.org/10.1016/j.physrep.2012.01.001}{{\em
  Phys. Rept.} {\bfseries 513} (2012) 1--189},
\href{http://arxiv.org/abs/1106.2476}{{\ttfamily arXiv:1106.2476
  [astro-ph.CO]}}.

\bibitem{Cho:1975dh}
Y.~M. Cho, ``{Einstein Lagrangian as the Translational Yang-Mills
  Lagrangian},'' \href{http://dx.doi.org/10.1103/PhysRevD.14.2521}{{\em Phys.
  Rev.} {\bfseries D14} (1976) 2521}.
[,393(1975)].

\bibitem{Ferraro:2006jd}
R.~Ferraro and F.~Fiorini, ``{Modified teleparallel gravity: Inflation without
  inflaton},'' \href{http://dx.doi.org/10.1103/PhysRevD.75.084031}{{\em Phys.
  Rev.} {\bfseries D75} (2007) 084031},
\href{http://arxiv.org/abs/gr-qc/0610067}{{\ttfamily arXiv:gr-qc/0610067
  [gr-qc]}}.

\bibitem{Hehl:2012pi}
F.~W. Hehl, ``{Gauge Theory of Gravity and Spacetime},''
  \href{http://dx.doi.org/10.1007/978-1-4939-3210-8_5}{{\em Einstein Stud.}
  {\bfseries 13} (2017) 145--169},
\href{http://arxiv.org/abs/1204.3672}{{\ttfamily arXiv:1204.3672 [gr-qc]}}.

\bibitem{Robertson1949}
H.~P. Robertson, ``{Postulate versus Observation in the Special Theory of
  Relativity},'' \href{http://dx.doi.org/10.1103/RevModPhys.21.378}{{\em Rev.\
  Mod.\ Phys.} {\bfseries 21} (1949) 378}.

\bibitem{Mansouri1977}
R.~Mansouri and R.~U. Sexl, ``{A Test Theory of Special Relativity: {I.
  S}imultaneity and Clock Synchronization},''
  \href{http://dx.doi.org/10.1007/BF00762634}{{\em Gen.\ Rel.\ Grav.}
  {\bfseries 8} (1977) 497}.

\bibitem{Mansouri1977a}
R.~Mansouri and R.~U. Sexl, ``{A Test Theory of Special Relativity: {II. F}irst
  Order Tests},'' \href{http://dx.doi.org/10.1007/BF00762635}{{\em Gen.\ Rel.\
  Grav.} {\bfseries 8} (1977) 515}.

\bibitem{Mansouri1977b}
R.~Mansouri and R.~U. Sexl, ``{A Test Theory of Special Relativity: {III.
  S}econd-Order Tests},'' \href{http://dx.doi.org/10.1007/BF00759585}{{\em
  Gen.\ Rel.\ Grav.} {\bfseries 8} (1977) 809}.

\bibitem{Bogoslovsky:1999pp}
G.~{\relax Yu}. Bogoslovsky and H.~F. Goenner, ``{Finslerian spaces possessing
  local relativistic symmetry},''
  \href{http://dx.doi.org/10.1023/A:1026786505326}{{\em Gen. Rel. Grav.}
  {\bfseries 31} (1999) 1565--1603},
\href{http://arxiv.org/abs/gr-qc/9904081}{{\ttfamily arXiv:gr-qc/9904081
  [gr-qc]}}.

\bibitem{Pfeifer:2011xi}
C.~Pfeifer and M.~N.~R. Wohlfarth, ``{Finsler geometric extension of Einstein
  gravity},'' \href{http://dx.doi.org/10.1103/PhysRevD.85.064009}{{\em Phys.
  Rev.} {\bfseries D85} (2012) 064009},
\href{http://arxiv.org/abs/1112.5641}{{\ttfamily arXiv:1112.5641 [gr-qc]}}.

\bibitem{Gielen:2012fz}
S.~Gielen and D.~K. Wise, ``{Lifting General Relativity to Observer Space},''
  \href{http://dx.doi.org/10.1063/1.4802878}{{\em J. Math. Phys.} {\bfseries
  54} (2013) 052501},
\href{http://arxiv.org/abs/1210.0019}{{\ttfamily arXiv:1210.0019 [gr-qc]}}.

\bibitem{Javaloyes:2018lex}
M.~A. Javaloyes and M.~Sánchez, ``{On the definition and examples of cones and
  Finsler spacetimes},''
\href{http://arxiv.org/abs/1805.06978}{{\ttfamily arXiv:1805.06978 [math.DG]}}.

\bibitem{Hohmann:2015pva}
M.~Hohmann, ``{Spacetime and observer space symmetries in the language of
  Cartan geometry},'' \href{http://dx.doi.org/10.1063/1.4961152}{{\em J. Math.
  Phys.} {\bfseries 57} no.~8, (2016) 082502},
\href{http://arxiv.org/abs/1505.07809}{{\ttfamily arXiv:1505.07809 [math-ph]}}.

\bibitem{Hohmann:2013fca}
M.~Hohmann, ``{Extensions of Lorentzian spacetime geometry: From Finsler to
  Cartan and vice versa},''
  \href{http://dx.doi.org/10.1103/PhysRevD.87.124034}{{\em Phys. Rev.}
  {\bfseries D87} no.~12, (2013) 124034},
\href{http://arxiv.org/abs/1304.5430}{{\ttfamily arXiv:1304.5430 [gr-qc]}}.

\bibitem{Gubitosi:2013rna}
G.~Gubitosi and F.~Mercati, ``{Relative Locality in $\kappa$-Poincar\'e},''
  \href{http://dx.doi.org/10.1088/0264-9381/30/14/145002}{{\em
  Class.Quant.Grav.} {\bfseries 30} (2013) 145002},
\href{http://arxiv.org/abs/1106.5710}{{\ttfamily arXiv:1106.5710 [gr-qc]}}.

\bibitem{Barcaroli:2017gvg}
L.~Barcaroli, L.~K. Brunkhorst, G.~Gubitosi, N.~Loret, and C.~Pfeifer,
  ``{Curved spacetimes with local $\kappa$-Poincaré dispersion relation},''
  \href{http://dx.doi.org/10.1103/PhysRevD.96.084010}{{\em Phys. Rev.}
  {\bfseries D96} no.~8, (2017) 084010},
\href{http://arxiv.org/abs/1703.02058}{{\ttfamily arXiv:1703.02058 [gr-qc]}}.

\bibitem{Raetzel:2010je}
D.~Raetzel, S.~Rivera, and F.~P. Schuller, ``{Geometry of physical dispersion
  relations},'' \href{http://dx.doi.org/10.1103/PhysRevD.83.044047}{{\em Phys.
  Rev.} {\bfseries D83} (2011) 044047},
\href{http://arxiv.org/abs/1010.1369}{{\ttfamily arXiv:1010.1369 [hep-th]}}.

\bibitem{Rubilar:2007qm}
G.~F. Rubilar, ``{Linear pre-metric electrodynamics and deduction of the light
  cone},''
  \href{http://dx.doi.org/10.1002/1521-3889(200211)11:10/11<717::AID-ANDP717>3.0.CO;2-6}{{\em
  Annalen Phys.} {\bfseries 11} (2002) 717--782},
\href{http://arxiv.org/abs/0706.2193}{{\ttfamily arXiv:0706.2193 [gr-qc]}}.

\bibitem{Punzi:2007di}
R.~Punzi, M.~N.~R. Wohlfarth, and F.~P. Schuller, ``{Propagation of light in
  area metric backgrounds},''
  \href{http://dx.doi.org/10.1088/0264-9381/26/3/035024}{{\em Class. Quant.
  Grav.} {\bfseries 26} (2009) 035024},
\href{http://arxiv.org/abs/0711.3771}{{\ttfamily arXiv:0711.3771 [hep-th]}}.

\bibitem{Barcaroli:2015xda}
L.~Barcaroli, L.~K. Brunkhorst, G.~Gubitosi, N.~Loret, and C.~Pfeifer,
  ``{Hamilton geometry: Phase space geometry from modified dispersion
  relations},'' \href{http://dx.doi.org/10.1103/PhysRevD.92.084053}{{\em Phys.
  Rev.} {\bfseries D92} no.~8, (2015) 084053},
\href{http://arxiv.org/abs/1507.00922}{{\ttfamily arXiv:1507.00922 [gr-qc]}}.

\bibitem{Schuller:2016onj}
M.~Düll, F.~P. Schuller, N.~Stritzelberger, and F.~Wolz, ``{Gravitational
  closure of matter field equations},''
  \href{http://dx.doi.org/10.1103/PhysRevD.97.084036}{{\em Phys. Rev.}
  {\bfseries D97} no.~8, (2018) 084036},
\href{http://arxiv.org/abs/1611.08878}{{\ttfamily arXiv:1611.08878 [gr-qc]}}.

\bibitem{ThePrincipalOfRelativity}
A.~Einstein, H.~A. Lorentz, M.~Hermann, and H.~Weyl, {\em The Principle of
  Relativity: A Collection of Original Memoirs on the Special and General
  Theory of Relativity}.
\newblock Dover Publications, NewYork, 1952.

\bibitem{Gurlebeck:2018nme}
N.~Gürlebeck and C.~Pfeifer, ``{Observers’ measurements in premetric
  electrodynamics: Time and radar length},''
  \href{http://dx.doi.org/10.1103/PhysRevD.97.084043}{{\em Phys. Rev.}
  {\bfseries D97} no.~8, (2018) 084043},
\href{http://arxiv.org/abs/1801.07724}{{\ttfamily arXiv:1801.07724 [gr-qc]}}.

\bibitem{smolin1986nature}
L.~Smolin, ``On the nature of quantum fluctuations and their relation to
  gravitation and the principle of inertia,'' {\em Classical and Quantum
  Gravity} {\bfseries 3} no.~3, (1986) 347.

\bibitem{Hoehn:2014vua}
P.~A. H\"ohn and M.~P. M\"uller, ``{An operational approach to spacetime
  symmetries: Lorentz transformations from quantum communication},''
  \href{http://dx.doi.org/10.1088/1367-2630/18/6/063026}{{\em New J. Phys.}
  {\bfseries 18} no.~6, (2016) 063026},
\href{http://arxiv.org/abs/1412.8462}{{\ttfamily arXiv:1412.8462 [quant-ph]}}.

\bibitem{Hoehn:2017gst}
P.~A. H\"ohn, ``{Reflections on the information paradigm in quantum and
  gravitational physics},''
  \href{http://dx.doi.org/10.1088/1742-6596/880/1/012014}{{\em J. Phys. Conf.
  Ser.} {\bfseries 880} no.~1, (2017) 012014},
\href{http://arxiv.org/abs/1706.06882}{{\ttfamily arXiv:1706.06882 [hep-th]}}.

\bibitem{Hardy:2018kbp}
L.~Hardy, ``{The Construction Interpretation: Conceptual Roads to Quantum
  Gravity},''
\href{http://arxiv.org/abs/1807.10980}{{\ttfamily arXiv:1807.10980
  [quant-ph]}}.

\bibitem{Rovelli:2004tv}
C.~Rovelli, {\em {Quantum Gravity}}.
\newblock Cambridge University Press,
2004.
\newblock

\bibitem{mercati2018shape}
F.~Mercati, {\em Shape Dynamics: Relativity and Relationalism}.
\newblock Oxford University Press, 2018.

\bibitem{Barbour295}
J.~Barbour and B.~Bertotti, ``Mach's principle and the structure of dynamical
  theories,'' \href{http://dx.doi.org/10.1098/rspa.1982.0102}{{\em Proceedings
  of the Royal Society of London A: Mathematical, Physical and Engineering
  Sciences} {\bfseries 382} no.~1783, (1982) 295--306}.

\bibitem{Pfeifer:2014eva}
C.~Pfeifer, ``{The tangent bundle exponential map and locally autoparallel
  coordinates for general connections on the tangent bundle with application to
  Finsler geometry},'' \href{http://dx.doi.org/10.1142/S0219887816500237}{{\em
  Int. J. Geom. Meth. Mod. Phys.} {\bfseries 13} no.~03, (2016) 1650023},
\href{http://arxiv.org/abs/1406.5413}{{\ttfamily arXiv:1406.5413 [math-ph]}}.

\bibitem{Minguzzi2016}
E.~Minguzzi, ``Special coordinate systems in pseudo-{F}insler geometry and the
  equivalence principle,''
  \href{http://dx.doi.org/https://doi.org/10.1016/j.geomphys.2016.12.013}{{\em
  Journal of Geometry and Physics} {\bfseries 114} (2017) 336 -- 347}.

\bibitem{Jacobson:2015hqa}
T.~Jacobson, ``{Entanglement Equilibrium and the Einstein Equation},''
  \href{http://dx.doi.org/10.1103/PhysRevLett.116.201101}{{\em Phys. Rev.
  Lett.} {\bfseries 116} no.~20, (2016) 201101},
\href{http://arxiv.org/abs/1505.04753}{{\ttfamily arXiv:1505.04753 [gr-qc]}}.

\bibitem{Jafarian}
A.~A. Jafarian and A.~Sourour, ``Spectrum-preserving linear maps,'' {\em
  Journal of functional analysis} {\bfseries 66} no.~2, (1986) 255--261.

\bibitem{MarvianSpekkens}
I.~Marvian and R.~W. Spekkens, ``How to quantify coherence: Distinguishing
  speakable and unspeakable notions,'' {\em Physical Review A} {\bfseries 94}
  no.~5, (2016) 052324.

\bibitem{Bartlett:2007zz}
S.~D. Bartlett, T.~Rudolph, and R.~W. Spekkens, ``{Reference frames,
  superselection rules, and quantum information},''
\href{http://dx.doi.org/10.1103/RevModPhys.79.555}{{\em Rev. Mod. Phys.}
  {\bfseries 79} (2007) 555--609}.

\bibitem{BratteliRobinson}
O.~Bratteli and D.~W. Robinson, {\em Operator Algebras and Quantum Statistical
  Mechanics: Volume 1: C*-and W*-Algebras. Symmetry Groups. Decomposition of
  States}.
\newblock Springer Science \& Business Media, 2012.

\bibitem{Loewy}
R.~Loewy, ``Linear transformations which preserve or decrease rank,'' {\em
  Linear Algebra and Its Applications} {\bfseries 121} (1989) 151--161.

\bibitem{Palmer:2011bt}
M.~C. Palmer, M.~Takahashi, and H.~F. Westman, ``{Localized qubits in curved
  spacetimes},'' \href{http://dx.doi.org/10.1016/j.aop.2011.10.009}{{\em Annals
  Phys.} {\bfseries 327} (2012) 1078--1131},
\href{http://arxiv.org/abs/1108.3896}{{\ttfamily arXiv:1108.3896 [quant-ph]}}.

\bibitem{Palmer:2013lva}
M.~C. Palmer, M.~Takahashi, and H.~F. Westman, ``{WKB analysis of relativistic
  Stern-Gerlach measurements},''
  \href{http://dx.doi.org/10.1016/j.aop.2013.05.017}{{\em Annals Phys.}
  {\bfseries 336} (2013) 505--516},
\href{http://arxiv.org/abs/1208.6434}{{\ttfamily arXiv:1208.6434 [quant-ph]}}.

\bibitem{wald}
R.~Wald, {\em {General Relativity}}.
\newblock Chicago University Press, 1984.

\bibitem{geatti2012some}
L.~Geatti and M.~Moskowitz, ``Some transitive linear actions of real simple lie
  groups,'' {\em Journal of Lie Theory} {\bfseries 22} no.~1, (2012) 155--161.

\bibitem{onishchik1993lie}
A.~L. Onishchik, ``Lie groups and lie algebras i, encyclopaedia of mathematical
  sciences, vol. 20,'' 1993.

\bibitem{Jacobson:2005bg}
T.~Jacobson, S.~Liberati, and D.~Mattingly, ``{Lorentz violation at high
  energy: Concepts, phenomena and astrophysical constraints},''
  \href{http://dx.doi.org/10.1016/j.aop.2005.06.004}{{\em Annals Phys.}
  {\bfseries 321} (2006) 150--196},
\href{http://arxiv.org/abs/astro-ph/0505267}{{\ttfamily arXiv:astro-ph/0505267
  [astro-ph]}}.

\bibitem{Josset:2016vrq}
T.~Josset, A.~Perez, and D.~Sudarsky, ``{Dark Energy from Violation of Energy
  Conservation},'' \href{http://dx.doi.org/10.1103/PhysRevLett.118.021102}{{\em
  Phys. Rev. Lett.} {\bfseries 118} no.~2, (2017) 021102},
\href{http://arxiv.org/abs/1604.04183}{{\ttfamily arXiv:1604.04183 [gr-qc]}}.

\bibitem{Perez:2017krv}
A.~Perez and D.~Sudarsky, ``{Dark energy from quantum gravity discreteness},''
\href{http://arxiv.org/abs/1711.05183}{{\ttfamily arXiv:1711.05183 [gr-qc]}}.

\bibitem{Perez:2018wlo}
A.~Perez, D.~Sudarsky, and J.~D. Bjorken, ``{A microscopic model for an
  emergent cosmological constant},''
\href{http://arxiv.org/abs/1804.07162}{{\ttfamily arXiv:1804.07162 [gr-qc]}}.

\bibitem{Miron}
R.~Miron, H.~Dragos, H.~Shimada, and S.~Sabau, {\em The geometry of Hamilton
  and Lagrange Spaces}.
\newblock Kluwer Academic, 2001.

\bibitem{Pfeifer:2011tk}
C.~Pfeifer and M.~N.~R. Wohlfarth, ``{Causal structure and electrodynamics on
  Finsler spacetimes},''
  \href{http://dx.doi.org/10.1103/PhysRevD.84.044039}{{\em Phys. Rev.}
  {\bfseries D84} (2011) 044039},
\href{http://arxiv.org/abs/1104.1079}{{\ttfamily arXiv:1104.1079 [gr-qc]}}.

\bibitem{Beem}
J.~K. Beem, ``Indefinite {F}insler spaces and timelike spaces,'' {\em Can. J.
  Math.} {\bfseries 22} (1970) 1035.

\bibitem{Letizia:2016lew}
M.~Letizia and S.~Liberati, ``{Deformed relativity symmetries and the local
  structure of spacetime},''
  \href{http://dx.doi.org/10.1103/PhysRevD.95.046007}{{\em Phys. Rev.}
  {\bfseries D95} no.~4, (2017) 046007},
\href{http://arxiv.org/abs/1612.03065}{{\ttfamily arXiv:1612.03065 [gr-qc]}}.

\bibitem{Perlick}
V.~Perlick, \href{http://dx.doi.org/10.1007/3-540-46662-2}{{\em Ray Optics,
  {Fermat's} Principle, and Applications to General Relativity}}.
\newblock No.~61 in Lecture Notes in Physics. Springer, 2000.

\bibitem{Fewster:2017mtt}
C.~J. Fewster, C.~Pfeifer, and D.~Siemssen, ``{Quantum energy inequalities in
  premetric electrodynamics},''
  \href{http://dx.doi.org/10.1103/PhysRevD.97.025019}{{\em Phys. Rev.}
  {\bfseries D97} no.~2, (2018) 025019},
\href{http://arxiv.org/abs/1709.01760}{{\ttfamily arXiv:1709.01760 [math-ph]}}.

\bibitem{Lukierski:1991pn}
J.~Lukierski, H.~Ruegg, A.~Nowicki, and V.~N. Tolstoi, ``{Q deformation of
  Poincare algebra},''
\href{http://dx.doi.org/10.1016/0370-2693(91)90358-W}{{\em Phys.Lett.}
  {\bfseries B264} (1991) 331--338}.

\bibitem{Majid:1994cy}
S.~Majid and H.~Ruegg, ``{Bicrossproduct structure of kappa Poincare group and
  noncommutative geometry},''
  \href{http://dx.doi.org/10.1016/0370-2693(94)90699-8}{{\em Phys.Lett.}
  {\bfseries B334} (1994) 348--354},
\href{http://arxiv.org/abs/hep-th/9405107}{{\ttfamily arXiv:hep-th/9405107
  [hep-th]}}.

\bibitem{bourbaki1975elements}
N.~Bourbaki, {\em Elements of Mathematics: Lie Groups and Lie Algebras. Pt. 1.
  Chapters 1-3}.
\newblock Hermann, 1975.

\bibitem{king1981some}
O.~King, ``On some maximal subgroups of the classical groups,'' {\em Journal of
  Algebra} {\bfseries 68} no.~1, (1981) 109--120.

\end{thebibliography}\endgroup

\end{document}